\documentclass{lmcs}
\pdfoutput=1

\usepackage{lastpage}
\renewcommand{\dOi}{14(3:4)2018}
\lmcsheading{}{1--\pageref{LastPage}}{}{}%
{Jan.~02,~2017}{Jul.~26,~2018}{}

\usepackage{hyperref}
\hypersetup{hidelinks}
\newcommand{\rightsquigstar}{\rightsquigarrow^{*}}
\usepackage{theoremref}
 \usepackage{array} \usepackage{graphicx}
 \usepackage{amssymb} \usepackage{amsmath} \usepackage{stmaryrd}

\newcommand{\fctrans}[1]{\rightsquigarrow_{#1}}
\newcommand{\lkappa}{<\!\!\kappa}

\newcommand{\leqlambda}{\leqslant\!\!\lambda}
\newcommand{\llambda}{<\!\!\lambda}
\newcommand{\On}{\mathsf{Ord}}

\newcommand{\zat}[1]{0_{#1}}
\newcommand{\brtree}[1]{C}
\newcommand{\prune}[1]{D}

\newcommand{\biginfty}{\infty\hspace{-0.95em} \infty}
\newcommand{\subbiginfty}{\infty\hspace{-0.825em} \infty}
\newcommand{\supof}[1]{\mathrm{sup}({#1})}

\newcommand{\emptytup}{()}

\newcommand{\nats}{\ensuremath{\mathbb{N}}}

\newcommand{\alzero}{f}

\newcommand{\alone}[1]{\alpha_i}

\newcommand{\betone}[1]{\beta_i}
\newcommand{\bitsen}[1]{\overline{#1}}
\newcommand{\lbitsen}[1]{\overline{#1}}

\newcommand{\von}[1]{\vono{#1}} 
\newcommand{\vonset}[1]{{\vonoset{#1}}}
\newcommand{\vonsetset}[1]{{\vonosetset{#1}}}
\newcommand{\vono}[1]{v_{#1}} 
\newcommand{\vonoset}[1]{u_{#1}}
\newcommand{\vonosetset}[1]{t_{#1}}
\newcommand{\limp}[1]{\mathrm{LP}({#1})}     
\newcommand{\nump}[1]{\mathrm{FP}({#1})}  

\newcommand{\pset}{\mathcal{P}}

\newcommand{\psetfin}{\mathcal{P}_{\mathsf{f}}}

\newcommand{\psetcount}{\mathcal{P}_{\mathsf{c}}}

\newcommand{\pco}[2]{p^{#1}_{{#2}}}

\newcommand{\pcop}[2]{p_{#1}{#2}}
\newcommand{\pcote}[2]{p_{#1}{#2}}
\newcommand{\pcofc}[3]{p_{#2}{#3}}
\newcommand{\pfc}[2]{{#2}}

\newcommand{\bits}{\setbr{0,1}}

\newcommand{\simi}[1]{\sim_{#1}}
 \usepackage[matrix,arrow,curve]{xy}

\newcommand{\ords}{\mathsf{Ord}}

\newcommand{\nuord}[2]{\nu^{(#1)}{#2}}
\newcommand{\muord}[2]{\mu^{(#1)}{#2}}
\newcommand{\nucoalg}[2]{\nu^{(#1)}_{\mathrm{coalg}}{#2}}

\newcommand{\nuordp}[1]{\nuord{#1}{\pset}}
\newcommand{\nuordf}[1]{\nuord{#1}{F}}
\newcommand{\nucoalgf}[1]{\nucoalg{#1}{F}}

\newcommand{\nuordpc}[1]{\nuord{#1}{\psetcount}}

\newcommand{\nuordpf}[1]{\nuord{#1}{\psetfin}}
\newcommand{\nuordpk}[1]{\nuord{#1}{\psetkappa}}
\newcommand{\nucoalgpk}[1]{\nucoalg{#1}{\psetkappa}}
\newcommand{\nuconn}[2]{\nu^{({#2}\geqslant{#1})}}
\newcommand{\ddiag}[2]{D_{{#2}\geqslant{#1}}}

\newcommand{\psetsub}[1]{\pset_{{#1}}}
\newcommand{\psetkappa}{\psetsub{<\kappa}}
\newcommand{\psetlambda}{\psetsub{<\lambda}}
\newcommand{\psetpluskappa}{\pset^{+}_{<\kappa}}
\newcommand{\pred}[1]{\mathsf{pred}({#1})}
\newcommand{\predm}[2]{\mathsf{pred}^{#1}({#2})}
\newcommand{\eqdef}{\stackrel{\mbox{\rm {\tiny def}}}{=}}
\newcommand{\setbr}[1]{\{{#1}\}}
\newenvironment{spaceout}[1]{\begin{displaymath}\setlength{\extrarowheight}{3pt}\begin{array}{#1}}{\end{array}\setlength{\extrarowheight}{0pt}\end{displaymath} \noindent}

\newcommand{\set}{\mathbf{Set}}
\newcommand{\op}{^{\mathsf{op}}}
\newcommand{\betwixt }{\hspace{1em}}
\newcommand{\ccons}{\; : \;}
\newcommand{\id}{\mathsf{id}}

\begin{document}
\title{A Ghost at $\omega_1$}
\author{Paul Blain Levy}
\address{University of Birmingham}
\email{P.B.Levy@cs.bham.ac.uk}

\begin{abstract}
  In the final chain of the countable powerset functor, we show that the set at index $\omega_1$, regarded as a transition system, is not strongly extensional because it contains a ``ghost'' element that has no successor even though its component at each successor index is inhabited.  The method, adapted from a construction of Forti and Honsell, also gives ghosts at larger ordinals in the final chain of other subfunctors of the powerset functor.  This leads to a precise description of which sets in these final chains are strongly extensional.
\end{abstract}

\maketitle

\bibliographystyle{alpha}

%Dependent Choice is assumed throughout.  Axiom of Choice is indicated.

\emph{For Jirka Ad\'{a}mek on his 70th birthday, with thanks for his many contributions and the inspiration he has given to our community.}

\section{Introduction}

Initial algebras and final coalgebras of endofunctors are important in many areas of mathematics and computer science.  One versatile way of constructing an initial algebra of $F$ is to form the \emph{initial chain}~\cite{Adamek:freealgsautomatareal}, a transfinite sequence of objects $(\muord{i}{F})_{i \in \On}$ where $\ords$ denotes the class of ordinals.  We form the initial chain by applying $F$ at each successor ordinal and taking a colimit at each limit ordinal.   If it stabilizes, it yields an initial algebra $\mu F$.  Dually~\cite{Barr:terminalcoalg}, we form the \emph{final chain}  $(\nuord{i}{F})_{i \in \On}$,  by applying $F$ at each successor ordinal and taking a limit at each limit ordinal. If it stabilizes, it yields a final coalgebra $\nu F$.

These constructions make sense for any endofunctor on any category, provided the relevant colimits or limits exist.  But certain endofunctors on $\set$ have received particular attention: the powerset functor ($\pset$) and its subfunctors, notably finite powerset ($\psetfin$) and countable powerset ($\psetcount$), which send a set to its set of finite or countable subsets respectively.  That is because these functors have many applications, e.g.\ nondeterministic automata, the semantics of nondeterministic programs and the foundations of set theory.  For cardinality reasons, the powerset functor has no initial algebra or final coalgebra; when we refer to $\mu \pset$ or $\nu \pset$, these are proper classes.  By contrast, $\psetfin$ and $\psetcount$ do have a (small) initial algebra and final coalgebra.

Now the initial chains of $\psetfin$ and $\psetcount$ are easy to understand: each is an increasing sequence of subsets of the initial algebra.  But the final chains are more subtle.  Their form was established in~\cite{AdamekLevyMiliusMossSousa:saturate,Worrell:finalseq}.
\begin{itemize}
\item $\nuordpf{\omega}$ consists of the final
  coalgebra and some extra elements.  The next $\omega$ steps in the final chain of $\psetfin$ remove
  these extra elements, and the chain stabilizes at $\omega + \omega$.
\item $\nuordpc{\omega_1}$ consists of the final
  coalgebra and some extra elements.  The next $\omega$ steps in the final chain of $\psetcount$ remove
  these extra elements, and the chain stabilizes at $\omega_{1} + \omega$.
\end{itemize}
These descriptions may give the impression that the final chains of $\psetfin$ and $\psetcount$ are essentially similar.  However, they tell only part of the story.  For $\nuord{\omega}{\psetfin}$, as well as being a \emph{superset} of $\nu\psetfin$, is also a \emph{subset} of $\nu\pset$.  This is because it can be endowed with a transition relation, setting $x \rightsquigarrow y$ when $y_n \in x_{n+1}$ for all $n \in \nats$, that makes it into a \emph{strongly extensional} transition system, i.e.\ one where bisimilarity is equality.  Specifically, Worrell~\cite{Worrell:finalseq} characterized  $\nuord{\omega}{\psetfin}$ as the set of \emph{compactly branching} elements\footnote{Elements of $\nu \pset$ are represented in~\cite{AdamekLevyMiliusMossSousa:saturate,Schwencke:coequational,Worrell:finalseq} as strongly extensional trees modulo isomorphism.  Those trees are unrelated to the trees used in this paper.} of $\nu \pset$, also studied in~\cite{Abramsky:cookfinnwfset,KupkeKurzVenema:stonecoalg}.  While all these elements are compactly branching, they are not all finitely (or even countably) branching, and the extra $\omega$ steps are required to remove those that are not, so that only the finitely branching elements remain.

Is there a similar story for $\nuord{\omega_1}{\psetcount}$?  The question is asked in~\cite{AdamekLevyMiliusMossSousa:saturate}.  This set too can be endowed with a transition relation, setting  $x \rightsquigarrow y$ when $y_i \in x_{i+1}$ for all $i < \omega_1$.  But the resulting transition system is not strongly extensional.  We demonstrate this by giving two distinct elements that are ``dead'' in the sense of having no successor.  One of these is the expected dead element: each component at a successor index is empty.  The other's components at successor index are all inhabited.  Since the latter ``appears to be alive'', yet is dead, we call it a \emph{ghost}.

Thus we see a significant difference between $\nuord{\omega}{\psetfin}$ and $\nuord{\omega_1}{\psetcount}$.  An element of $\nuord{\omega}{\psetfin}$ may lie outside the final coalgebra of $\psetfin$, but only because it is not finitely branching.   By contrast, an element of  $\nuord{\omega_1}{\psetcount}$ may lie outside the final coalgebra of $\psetcount$ despite having no successors at all.

\subsection*{Structure of Paper}

Before introducing the final chain, Section~\ref{sect:prelim} gives preliminaries on transitions systems, cardinals, and the sequence of approximants to bisimilarity.  Section~\ref{sect:finalchain} introduces the final chain of a general functor, and in particular of the powerset and restricted powerset functors.  It also introduces the notions of channel and range that help us to understand these final chains.  

The main narrative begins in Section~\ref{sect:fctrans}, which endows the final chain of $\pset$ with the structure of a transition system and describe its basic properties.  The rest of the paper is devoted to studying this system.  

In Section~\ref{sect:count} our focus is on those properties that hold specifically at \emph{countable} ordinals.  This leads us in Section~\ref{sect:ghostomegaone} to the question of whether these properties---most importantly, strong extensionality---hold at $\omega_1$.  We prove that there is a ghost and deduce that these properties do not hold. 

The ghost is obtained by adapting a method in~\cite{FortiHonsell:modselfdescript}.  The next two sections give further results that showcase the power of this method.  
\begin{itemize}
\item Section~\ref{sect:numbersucc} answers the question: how many successors may an element of the final chain of $\psetcount$ have?
\item Section~\ref{sect:beyond} looks at larger subfunctors of $\pset$ corresponding to subsets of larger cardinality, and their final chain beyond $\omega_1$.  Remarkably, we see a sharp division: at each limit ordinal $i$, either the $i$th set is strongly extensional or there is a ghost.  We thus revisit each of the properties we initially proved at countable ordinals, including surjectivity of connecting maps, and completely characterize the ordinals at which they hold.
\end{itemize}
We end with a discussion of related work in Section~\ref{sect:related} and conclusions in Section~\ref{sect:conclusions}.

%  gives the central definition: the final chain viewed as a transition system.  It lists some basic properties and examples.  Then Section~\ref{sect:count} goes through the properties, mostly adapted from~\cite{AdamekLevyMiliusMossSousa:saturate}, that hold at countable ordinals.  This is to emphasize the contrast with $\omega_1$, where these properties do not hold.  This is shown in Section~\ref{sect:ghostomegaone} by demonstrating the existence of a ghost.  Section~\ref{sect:numbersucc} applies the techniques developed to consider how many successors an element of the final chain may have.

% Section~\ref{sect:beyond} proceeds beyond $\omega_1$, not for the countable powerset functor but for powerset restricted to larger cardinals.  Every property that was proved to hold at countable ordinals is fully analyzed in this section, i.e.\ we  precisely identify the ordinals at which it holds.  In particular, we identify which connecting maps of the final chain are surjective.  This is all accomplished by proving the existence of more ghosts.

% Section~\ref{sect:related}  discusses the relationship between our results and related literature such as~\cite{AdamekLevyMiliusMossSousa:saturate,FortiHonsell:modselfdescript,LazicRoscoe:nonwellfound}.  Section~\ref{sect:conclusions} concludes.

\subsection*{Acknowledgements.}  I thank David Fernandez-Breton and Benedikt L\"{o}we for supplying the information in footnotes~\ref{foot:kurepa}--\ref{foot:treeprop}.

\section{Preliminaries} \label{sect:prelim}
\subsection{Pointed Systems}

A \emph{transition system} consists of a set $X$ and relation $\rightsquigarrow\, \subseteq X \times X$.  When $x \rightsquigarrow y$ we say that $y$ is a \emph{successor} of $x$.  A subset $U \subseteq X$ is a \emph{subsystem} when, for all $x \in U$, every successor of $x$ is in $U$. An \emph{embedding} of transition systems $e \colon M \to N$ is an isomorphism from $M$ to a subsystem of $N$.
  
A \emph{pointed system} consists of a transition system $M=(X,\rightsquigarrow)$ and an element $x\in X$.  % (Roughly speaking, a pointed system represents an element of $\nu\pset$.)  
Given a family of pointed systems $((M_i,x_i))_{i \in I}$, a \emph{parent} of the family is a pointed system $(N,y)$ equipped with embeddings $(e_i \colon M_i \to N)_{i \in I}$ such that the successors of $y$ are listed without repetition as $(e_i(x_i))_{i \in I}$.  A parent may be constructed as follows: if each $M_i$ is $(X_i,\rightsquigarrow_i)$, define
\begin{eqnarray*}
  N & \eqdef & (1+\sum_{i \in I}X_i, \rightsquigarrow) \betwixt \text{with  $\rightsquigarrow$ given in the evident way}  \\ 
  y & \eqdef & \mathsf{inl}\,* \\
  e_i & \colon & z \mapsto \mathsf{inr}\,(i,z)
\end{eqnarray*}
%with $\rightsquigarrow$ given in the evident way.

% \begin{defi} \labe{def:buildps}
% Let $((M_i,x_i))_{i \in I}$ be a family of pointed systems, where $M_i = (X_i, \rightsquigarrow_i)$.
% \begin{enumerate}
% \item Let $N$ be the transition system $(1+\sum_{i \in I}X_i, \rightsquigarrow)$, where $\mathsf{inl}\,*$ has successors $\mathsf{inr}\,(i,x_i)$ for all
%   $i \in I$, and $\mathsf{inr}\,(i,z)$ has successors
%   $\mathsf{inr}\,(i,w)$ for all successors $w$ of $z$.
% \item For each $i \in I$, let the embedding  $e_i \colon M_i \to N$ send $z \mapsto \mathsf{inr}\,\pair{i,z}$.
%  \item Let $y$ be $\mathsf{inl}\,*$.  %Note that its successors are $e_i(x)$ for all $i \in I$.
% \end{enumerate}
% The pointed system $(N,y)$ is called the \emph{parent} of the family  $((M_i,x_i))_{i \in I}$.  
% \end{defi}
% Note the key fact: the successors of $y$ are  $e_i(x)$ for all $i \in I$.

We generalize transition systems to coalgebras.  Let $F$ be an endofunctor on $\set$, e.g.\ the powerset functor $\pset$, which sends a set $X$ to $\pset X$ and a function  $f \colon A \to B$ to $U \mapsto \setbr{f(a) \mid a \in U}$.  An \emph{$F$-coalgebra} consists of a set $X$ and map $X \to FX$.  Thus a transition system corresponds to a $\pset$-coalgebra.   An \emph{$F$-coalgebra morphism} $(X,\zeta) \to (Y,\xi)$ is a map $f \ccons X \to Y$ such that
    \begin{math}
      \xymatrix{
X \ar[d]^-{\zeta} \ar[r]_{f} & Y \ar[d]_{\xi} \\
 FX \ar[r]^-{Ff} & FY
 % X \ar[r]^-{\zeta} \ar[d]_{f} & FX \ar[d]^{Ff} \\
 % Y \ar[r]_-{\xi} & FY
        }
    \end{math} commutes.  For example, an embedding of transition systems corresponds to an injective $\pset$-coalgebra morphism.   A \emph{pointed $F$-coalgebra} consists of an $F$-coalgebra $M=(X,\zeta)$ together with an element $x \in X$.  % A \emph{pointed $F$-coalgebra morphism} $(M,x) \to (N,y)$ is an $F$-coalgebra morphism $f \colon M \to N$ such that $f(x) = y$. 

We often write a pointed system or pointed coalgebra $(M,x)$ as just $x$, leaving $M$ implicit.

\subsection{Cardinals}

As usual, we identify a cardinal with the least ordinal of that cardinality. Thus we identify $\aleph_0$ with $\omega$ and $\aleph_1$ with $\omega_1$---the least uncountable ordinal.  The special symbol $\biginfty$ is treated as greater than every cardinal and ordinal.

Let $\kappa$ be either a cardinal or $\biginfty$. % (Whenever we treat  $\kappa$ as an ordinal, it is implicitly assumed that $\kappa < \biginfty$.)  
For a set $X$, we write $\psetkappa X$ for the set of its $\lkappa$-sized subsets.  This gives a subfunctor $\psetkappa$ of $\pset$.   In particular:
\begin{eqnarray*}  \psetfin &  = & \pset_{< \aleph_0} \\
 \psetcount & = & \pset_{< \aleph_1} \\
 \pset & = & \pset_{< \subbiginfty}
\end{eqnarray*}  
We also refer to the functor $\psetpluskappa$, that sends a set $X$ to the set of $\lkappa$-sized subsets that are inhabited.

A pointed system $(M,x)$ is \emph{$\lkappa$-branching} when $x \rightsquigstar y$ implies that $y$ has $\lkappa$-many successors.  This is equivalent to $(M',x)$ being a pointed $\psetkappa$-coalgebra, for some subsystem $M'$ of $M$.

Henceforth we assume that $\kappa$ is \emph{infinite and regular}.   (In particular $\aleph_0$ is regular, as are infinite successor cardinals such as $\aleph_1$.  But $\aleph_{\omega}$, for example, is not. We deem $\biginfty$ to be regular.)  This makes $\psetkappa$ a submonad of $\pset$, and it also gives the following.
\begin{prop} \label{prop:regular}
 For a $\lkappa$-sized set $A$, any decreasing sequence of subsets $(X_i)_{i < \kappa}$ is eventually constant.  In particular, if each  $X_i$ is inhabited then so is $\bigcap_{i < \kappa} X_i$. 
\end{prop}

\subsection{Approximants to bisimilarity}

The class of pointed systems is equipped with a decreasing sequence  $(\simi{i})_{i \in \On}$ of equivalence relations, defined as follows~\cite{Malitz:phd,Milner:ccsbook}.   
\begin{itemize}
\item $x \simi{i+1} x'$ when for every $x \rightsquigarrow y$ there is $x' \rightsquigarrow y'$ such that $y \simi{i} y'$, and \emph{vice versa}.
\item If $i$ is a limit, then $x \simi{i} x'$ when for all $j < i$ we have $x \simi{j} x'$.
\item We deem $0$ a limit, so $x \simi{0} x'$ always.
\end{itemize}
We say that $x$ and $x'$ are \emph{bisimilar}, written $x \sim x'$, when for all $i \in \On$ we have $x \simi{i} x'$. In particular, for any embedding $e \colon M \to N$ we have $(N,e(x)) \sim (M,x)$.  The following is a key tool for analyzing these relations.
\begin{defi}\label{def:von}
Pointed systems $(\von{i})_{i \in \On}$ are defined as follows: $\von{i}$ is a parent of $(\von{j})_{j < i}$. 
\end{defi}
Essentially Definition~\ref{def:von} is the von Neumann encoding of the ordinals. It is useful because of the following facts.
\begin{prop} \label{prop:von}
Let $i$ and $j$ be ordinals.
\begin{enumerate}
\item \label{item:vonsimi}  For an ordinal $k$, we have $\von{i} \simi{k} \von{j}$ iff either $i=j$ or $k \leqslant i,j$.
\item \label{item:vonsim} We have $\von{i} \sim \von{j}$ iff $i=j$.
\end{enumerate}
\end{prop}
\begin{proof}
  Part (\ref{item:vonsimi}) is by induction on $k$ and part (\ref{item:vonsim}) follows.
\end{proof}
\noindent The approximants to bisimilarity have the following properties.
\begin{prop}\label{prop:simkappa}\hfill
  \begin{enumerate}
\item  \label{item:simkappa}   Let $x$ be an $\lkappa$-branching pointed system and $y$ a pointed system.  Then $x \sim  y$ iff, for all $i < \kappa$, we have $x \simi{i} y$.
  \item  \label{item:simkappastrict} For $i < \kappa$, there are $\lkappa$-branching pointed systems $x$ and $y$ such that $x \simi{i} y$ but $x \not\simi{i+1} y$.
  \end{enumerate}
\end{prop}
\begin{proof}\leavevmode
\begin{enumerate}
\item  See~\cite{BarwiseMoss:vicious}[Lemma 11.13] following~\cite{vanglabbeek:bounded}, and also~\cite{AcetoIngolfsdottirLevySack:charform}[Theorem B.1(3)]. 
\item Take $\von{i}$ and $\von{i+1}$. \qedhere
\end{enumerate}
\end{proof}

\begin{prop} \label{prop:simkappasucc}
  Let $x$ be an $\lkappa$-branching pointed system and $y$ a pointed system.  Then $x$ has a successor $z$ such that  $z \sim y$ iff, for all $i < \kappa$, it has a successor $z$ such that $z \simi{i} y$.
\end{prop}
\begin{proof} \hfill
($\Rightarrow$) is obvious.  For ($\Leftarrow$), by Proposition~\ref{prop:regular}, $x$ has a successor $z$ such that $z \simi{j} y$ for all $j < \kappa$, i.e.\ $z \simi{k} y$.  Since $z$ is $\lkappa$-branching, Proposition~\ref{prop:simkappa}~(\ref{item:simkappa}) gives $z \sim y$. 
\end{proof}

\section{The Final Chain} \label{sect:finalchain}

\subsection{Constructing the Final Chain}

Our treatment of the final chain relies on the following.  Let $I$ be a well-ordered set.
\begin{defi}
  An \emph{inverse $I$-chain} is a functor $D \colon I^{\op} \to \set$.  Explicitly, it consists of
    \begin{itemize}
    \item for all $i \in I$, a set $D_i$, called the $i$-th
      \emph{level}
    \item for all $j \leqslant i \in I$, a \emph{connecting map}
      $\ddiag{j}{i} \ccons D_j \to D_i$
    \end{itemize}
    such that $\ddiag{i}{i} = \id_{D_i}$ for all $i \in I$, and
    \begin{math}
      \xymatrix{
        D_i \ar[r]^{\ddiag{j}{i}} \ar[dr]_{\ddiag{k}{i}} & D_{j} \ar[d]^{\ddiag{k}{j}} \\
        & D_{k} }
    \end{math}
    commutes for all $k \leqslant j \leqslant i \in I$.
\end{defi}
An inverse $I$-chain is a kind of tree~\cite{Jech:settheory,BirkedalMogelbergSchwinghammerStovring:topostrees,BizjakBirkedalMiculan:modcountnondet}.  It may be represented as a poset $A$ in which, for every $x \in A$, the set $\setbr{y \in A \mid x < y}$ is well-ordered with order-type less than that of $I$.  This intuitive picture gives rise to the following notions.
\begin{defi}
  Let $D$ be an inverse $I$-chain.
  \begin{enumerate}
  \item If $I$ has a least element $0$, then an element of $D_0$ is a \emph{root}.
\item For $j \leqslant i \in I$ and $x \in D_j$, an \emph{$i$-development} of $x$ is an element $y \in D_i$ such that
  $\ddiag{j}{i}y = x$.
 \item A \emph{full branch} of $D$ is  an indexed tuple
     $(x_i)_{i \in I}$ with $x_i \in D_i$ for all $i \in I$ and
     $\ddiag{j}{i}x_i= x_j$ for all $j \leqslant i \in I$.  (If $I$ is empty, the empty tuple $\emptytup$ is the sole full branch.)
  % \item A \emph{channel}  of $D$ consists of a subset $C_i \subseteq D_i$ for all $i\in I$, such that for all $j \leqslant i \in I$
  % \begin{itemize}
  % \item if $x \in C_i$ then $\ddiag{j}{i} x \in C_j$
  % \item every $y \in C_j$ has an $i$-development in $C_i$.
  % \end{itemize}
  % \item For any set $U$ of full branches of $D$, the \emph{range} of $U$ is the channel whose $i$th level is $\setbr{x_{i} \mid x \in U}$.
  \end{enumerate}
\end{defi}

We now proceed to define the final chain of an endofunctor $F$ on $\set$.
\begin{defi}~\cite{Adamek:freealgsautomatareal,Barr:terminalcoalg}  \label{def:finalseq}
%Let $\catc$ have (distinguished) limits.
%  \begin{enumerate}\item
The \emph{final chain} of $F$ is an inverse $\On$-chain, with $i$th level written $\nuordf{i}$ and
  connecting map written
  $\nuconn{j}{i}F \ccons \nuord{i}{F} \longrightarrow \nuord{j}{F}$ or
  $\nuconn{j}{i}$ for short.  These sets and maps are given as
  follows.
  \begin{itemize}\item
    $\nuordf{i+1} = F\nuordf{i}$.
  \item $\nuconn{j+1}{i+1} = F\nuconn{j}{i}$ for $j \leqslant i$.
  \item If $i$ is a limit then $\nuordf{i}$ is the limit of $(\nuordf{j})_{j < i }$, i.e.\ the set of full branches.  For
    $j < i$, the map $\nuconn{j}{i}$ is
    $\pi_j \ccons (x_j)_{j < i} \mapsto x_j$.
  \item Since we deem 0 a limit ordinal, we have
    $\nuordf{0} = \setbr{\emptytup}$.
  \end{itemize}
\end{defi}
Thus 
\begin{spaceout}{ccccc}
  \nuconn{j+1}{i+1}\psetkappa & \ccons & a & \mapsto & \setbr{\nuconn{j}{i}b \mid b \in a} 
\end{spaceout}%

We define a full branch $(\zat{i})_{i \in \On}$ through the final chain of $\psetkappa$, as follows:
\begin{itemize}
\item $\zat{i+1} \eqdef \emptyset$.
\item If $i$ is a limit then $\zat{i} \eqdef (\zat{j})_{j < i}$.
\end{itemize}
For $i > 0$, note that $\zat{i}$ is the unique $i$-development of $\zat{1}$.

\subsection{Coalgebra projections} \label{sect:coalgproj}

Again let $F$ be an endofunctor on $\set$.
\begin{defi}
 For an $F$-coalgebra $(X,\zeta)$ and ordinal $i$, the $i$th
    \emph{coalgebra projection}, written $\pco{(X,\zeta)}{i}$ or just $\pco{}{i}$, is the map $X \to \nuordf{i}$ given as follows\footnote{In the language of~\cite{CaprettaUustaluVene:corecursive}, the pair $(\nuordf{i},\nuconn{i}{i+1})$ is a \emph{corecursive $F$-algebra}, and $\pco{(X,\zeta)}{i}$ is the unique map from $(X,\zeta)$ to it.  The connecting maps are $F$-algebra morphisms.}.
    \begin{itemize}
    \item The map $\pco{}{i+1}$ is 
    \begin{math}
      \xymatrix{
        X \ar[r]^-{\zeta} & FX\ar[r]^-{F\pco{}{i}} & F\nuordf{i} 
      } 
    \end{math}
   \item For a limit $i$, the map $\pco{}{i}$ is   $(\pco{}{j})_{j < i}$.
    \end{itemize}
\end{defi}
We may summarize as follows.
\begin{itemize}
\item Each pointed $F$-coalgebra $x$ gives rise to a full branch $(\pcop{i}{x})_{i\in \On}$ through the final chain of $F$.
\item For any $F$-coalgebra morphism $f \colon (X,\zeta) \to (Y,\xi)$ and $x \in X$, the full branches  $(\pcop{i}{x})_{i\in \On}$ and  $(\pcop{i}{f(x)})_{i\in \On}$ are equal.
\end{itemize}
For a pointed system $x$, we have $\pcop{i+1}{x} = \setbr{\pcop{i}{y} \mid x \rightsquigarrow y}$.  The coalgebra projections are related to bisimilarity and its approximants as follows.
\begin{prop} \label{prop:kernelsim}
  Let $x$ and $y$ be pointed systems.
  \begin{enumerate}
  \item \label{item:kernelsimi} For any ordinal $i$, we have $x \simi{i} y$ iff $\pcop{i}{x} = \pcop{i}{y}$. 
  \item \label{item:kernelsim} We have $x \sim y$ iff the full branches  $(\pcop{i}{x})_{i\in \On}$ and  $(\pcop{i}{y})_{i\in \On}$ are equal.
  \end{enumerate}
\end{prop}
\begin{proof} Part~(\ref{item:kernelsimi}) is by induction on $i$ and part~(\ref{item:kernelsim}) follows. \end{proof}
% Proposition~\ref{prop:kernelsim} suggests the following defintions, for general $F$.
% \begin{spaceout}{rcl}
%   x  \simi{i} y & \iffdef & \pcop{i}{x} = \pcop{i}{y} \\
%   x \sim y & \iffdef & (\pcop{i}{x})_{i\in \On} = (\pcop{i}{y})_{i\in \On}
% \end{spaceout}%
% See~\cite{} for comparison between $\sim$ and other notions of bisimilarity.

\begin{defi}
  An element $a \in \nuordpf{i}$ is said to be \emph{$F$-coalgebraic} when it is of the form $\pcop{i}{x}$ for a pointed $F$-coalgebra $x$.  We write $\nucoalgf{i}$ for the set of $F$-coalgebraic elements of $\nuordf{i}$.
\end{defi}
Thus $\nucoalgpk{i}$ corresponds to the class of $\lkappa$-branching pointed systems modulo $\simi{i}$.  Note that $a \in \nuordpk{i+1}$ is $\psetkappa$-coalgebraic iff all its elements are, by the parent construction.
\begin{defi}
  Let $i$ be a limit.  An element  of $\nuordf{i}$ is \emph{$F$-Cauchy} when all of its components are $F$-coalgebraic.
\end{defi}
Thus the set of $F$-Cauchy elements is $\lim_{j<i} \nucoalgf{j}$.  Note that $F$-coalgebraic implies $F$-Cauchy, but the converse need not be the case.  (For example, for $F \colon X \mapsto 0$, the element $\emptytup \in \nuordf{0}$ is $F$-Cauchy but not $F$-coalgebraic.)  As explained in Section~\ref{sect:cauchyco}, the ``Cauchy'' terminology comes from the notion of Cauchy sequence.

% We accordingly define $\simi{i}$ on pointed $F$-coalgebras by saying $x \simi{i} y$ when $\pcop{i}{x} = \pcop{i}{y}$, and obtain notions of Cauchy $i$-sequence of pointed $F$-coalgebras.
% \begin{defi}\hfill
%   \begin{enumerate}
%     \item An element $a \in \nuordpf{i}$ is \emph{$F$-coalgebraic} when $a = \pcop{i}{x}$ for a pointed $F$-coalgebra $x$. 
%   \item For a limit ordinal $$, an element  $a \in \nuordpf{i}$ is \emph{$F$-Cauchy} when for all $j<i$ the component $a_j$ is $F$-coalgebraic. 
%   \end{enumerate}
% \end{defi}

\subsection{Channels and Ranges} \label{sect:channel}

The following notions are key to understanding the final chain of $\psetkappa$.
\begin{defi}
  Let $I$ be a well-ordered set, and $D$ an inverse $I$-chain.
  \begin{enumerate}
  \item A \emph{channel} through $D$ consists of a subset $C_i \subseteq D_i$ for all $i\in I$, such that for all $j \leqslant i \in I$
  \begin{itemize}
  \item if $x \in C_i$ then $\ddiag{j}{i} x \in C_j$
  \item every $y \in C_j$ has an $i$-development in $C_i$.
  \end{itemize}
  \item For any set $U$ of full branches of $D$, the \emph{range} of $U$ is the channel whose $i$th level is $\setbr{x_{i} \mid x \in U}$.
  \end{enumerate}
\end{defi}
Note the following:
\begin{itemize}
\item Any channel through $D$ with a root has all levels inhabited.
\item $\lim_{i \in I} \psetkappa D_i$ is the set of channels through $D$ with all levels $\lkappa$-sized.
\item The function $(\psetkappa \pi_i)_{i \in I} \colon \psetkappa \lim_{i \in I}D_i \to \lim_{i \in I} \psetkappa D_i$ sends a $\lkappa$-sized set of full branches to its range.
\end{itemize}

For general $F$, we have the following.
\begin{prop} \label{prop:psi}
Let $i$ be an ordinal.
  \begin{enumerate}
  \item The map $\Psi_i \eqdef (\nuconn{j}{i})_{j < i} \colon \nuordf{i} \cong  \lim_{j < i} F \nuordf{j}$ is an isomorphism (bijection).
  \item \label{item:psicomm} If $i$ is a limit, the following commutes:
\begin{displaymath}
    \xymatrix{
 \nuordf{i} & = & \lim_{j < i} \nuordf{j} \ar[d]_{\Psi_i}^{\cong} & & F \lim_{j < i} \nuordf{j}  \ar[ll]_{\nuconn{i}{i+1}} \ar[dll]^{(F\pi_{j})_{j < i}} & = & \nuordf{i+1} \\
& &  \lim_{j < i} F \nuordf{j}
}
  \end{displaymath}
  \end{enumerate}
\end{prop}
Consider the case $F = \psetkappa$.  For a limit $i$, the elements of $\nuordpk{i}$ are, by definition, the \emph{full branches} of the inverse chain $(\nuordpk{j})_{j < i}$.  But they correspond via $\Psi_i$ to the \emph{channels} through the same inverse chain.  We repeatedly move between these two viewpoints---full branches and channels---in this paper.  For example, the map $\nuconn{i}{i+1}$ is hard to grasp directly, but by Proposition~\ref{prop:psi}(\ref{item:psicomm}) it corresponds to $(\psetkappa\pi_j)_{j < i}$, which sends a set of full branches to its range.  

\subsection{Injectivity of Connecting Maps} \label{sect:stable}

A key question in the study of final chains (as with initial chains) is stabilization~\cite{TrnkovaAdamekKoubekReiterman:free,AdamekKoubek:terminal,AdamekTrnkova:manysorted,Worrell:finalseq,AdamekLevyMiliusMossSousa:saturate,AdamekPalm:steps}.
\begin{prop}  \label{prop:ifstable}
 Suppose the final chain of $F$ is \emph{stable} at $i$, i.e.\ $\nuconn{i}{i+1}$ is bijective.
  \begin{itemize}
  \item $\nuconn{j}{k}$ is bijective for all $k \geqslant j \geqslant i$.
  \item \label{item:stabfin} The $F$-coalgebra $(\nuordf{i}, (\nuconn{i}{i+1})^{-1})$ is final, with the unique coalgebra map from any coalgebra $M$ given by $p^{M}_i$.
  \item The $j$th projection from this final coalgebra is given by $\nuconn{j}{i}$ for $j \leqslant i$, and by $(\nuconn{i}{j})^{-1}$ for $j \geqslant i$.
  % \item The map $\pco{(\nuordf{i}, (\nuconn{i}{i+1})^{-1})}{j}$  is given by
  %   \begin{itemize}
  %   \item $\nuconn{j}{i}$ for $j \leqslant i$
  %   \item $(\nuconn{i}{j})^{-1}$ for $j \geqslant i$.
  %   \end{itemize}
    \qedhere
  \end{itemize}
\end{prop}
% For example, the final chain of $\psetpluskappa$ is stable at $1$ for $\kappa =1$ and at $0$ for $\kappa > 1$.
By Proposition~\ref{prop:psi}(\ref{item:psicomm}), for any limit $i$, if $F$ preserves limits of inverse $i$-chains then its final chain  is stable at $i$.  This fact is useful for  a \emph{polynomial} functor $F$, i.e.\ one of the form $X \mapsto \sum_{i \in I} X^{A_i}$.  Since it preserves limits of connected diagrams, its final chain is stable at $\omega$~\cite{Barr:terminalcoalg}.  %Furthermore $\nuordf{n}$ may be seen as the set of trees with stumps at level $n$.   
However, $\psetkappa$ requires a more subtle analysis, given in~\cite{Worrell:finalseq}, that we now reprise.

\begin{prop} \label{prop:allinj}
Let $F$ be an endofunctor on $\set$ that preserves injections and intersections (such as $\psetkappa$ or $\psetpluskappa$). Suppose  $\nuconn{i}{i+1}F$ is an injection.  Then  $\nuconn{j}{k}F$ is an injection for all  $k \geqslant j \geqslant i$, and the final chain is stable at $i+\omega$.
 \end{prop}

A functor $F$ is said to preserve the limit of an inverse $I$-chain $D$ up to an injection (surjection) when the map 
\begin{displaymath}
(F\pi_i)_{i \in I} \colon  F\lim_{i \in I} D_i \to \lim_{i \in I} FD_i
\end{displaymath}
is injective (surjective).

\begin{prop}~\cite{Worrell:finalseq} \label{prop:rangeinj} Let $\kappa <\biginfty$.  Then $\psetkappa$ and $\psetpluskappa$ preserve, up to an injection, the limits of any inverse $\kappa$-chain $D$.   Explicitly: if  $U$ and $V$ are $\lkappa$-sized sets of full branches through $D$ and have the same range, then $U=V$.
\end{prop}
\begin{proof}
  Let $a \in U$. For $i<\kappa$, let $V(i)$ be the set of $b \in V$ such that $b_i = a_i$.   Since $V$ is $\lkappa$-sized, $\bigcap_{i < \kappa}V(i)$ is inhabited by Proposition~\ref{prop:regular}, so $a \in V$.
\end{proof}
Thus $\nuconn{\kappa}{\kappa+1}(\psetkappa)$ is injective and the final chain of $\psetkappa$ is stable at $\kappa + \omega$.  As proved in~\cite{AdamekLevyMiliusMossSousa:saturate} and reprised below (Proposition~\ref{prop:noearlier}), it is not stable at any smaller ordinal.

We can now say precisely which connecting maps are injective.
\begin{prop}  \label{prop:injective}
  For ordinals $j \leqslant i$, the following are equivalent.
\begin{enumerate}
\item \label{item:injectivep} The connecting map $\nuconn{j}{i}\psetkappa$ is an injection.
\item \label{item:injectivec} Either $j=i$ or $j \geqslant \kappa$ (for $\kappa < \biginfty$).
\end{enumerate}
\end{prop}
\begin{proof}
(\ref{item:injectivec})$\Rightarrow$(\ref{item:injectivep}) follows from Propositions~\ref{prop:rangeinj} and~\ref{prop:allinj}.  For the converse, if $j<\kappa$, then Proposition~\ref{prop:simkappa}(\ref{item:simkappastrict}) shows that $\nuconn{j}{j+1} \psetkappa$ is not injective.
\end{proof}

\subsection{Surjectivity of connecting maps}

\begin{prop} \label{prop:projconnect} 
For any ordinal $i$, the following are equivalent.
  \begin{enumerate}
  \item \label{item:succsurj} The map $\nuconn{i}{i+1}F$ is surjective.
  % \item \label{item:somesuccsurj} $\nuconn{i}{j}$ is surjective for some $j > i$.
%  \item \label{item:somesuccsurj} $\nuordf{i,j}$ is surjective for some $j > i$.
% For every $a \in \nuordf{i}$ there is $j > i$ such that $a$ is in the range of $\nuordf{i,j}$.
  \item \label{item:allsuccsurj} For all $h \geqslant i$, the map $\nuconn{i}{h}F$ is surjective.
   % \item \label{item:finalsurj} The projection $p_{\nu F, i}$, where $\nu F$ is the large\footnote{Though not if we assume $F$ to be $\omega_{1}$-accessible.} final $F$-coalgebra, is surjective.
\item \label{item:somecorange} All elements of $\nuordf{i}$ are $F$-coalgebraic.
 \end{enumerate}
\end{prop}
\begin{proof}
  \begin{description}\item[(\ref{item:somecorange}) $\Rightarrow$ (\ref{item:allsuccsurj}) $\Rightarrow$ (\ref{item:succsurj})] Trivial.
\item[(\ref{item:succsurj}) $\Rightarrow$ (\ref{item:somecorange})] By the Axiom of Choice,  $\nuconn{i}{i+1}$ has a section $\zeta \ccons \nuordf{i} \to \nuordf{i+1}$.  Put $M =
    (\nuordf{i},\zeta)$. For all $j \leqslant i$, we show $\pco{M}{j} =
    \nuconn{j}{i}$ by induction on $j$.  The limit case is trivial.  For the successor case, $\pco{M}{j} = \nuconn{j}{i}$
      implies $F\pco{M}{j} = F\nuconn{j}{i} = \nuconn{j+1}{i+1}$ and hence
      \begin{displaymath}
        \xymatrix{
          \nuordf{i} \ar@/^1pc/[drr]^-{\pco{M}{j+1}} \ar[dr]_{\zeta} \ar@/_{1pc
}/[dd]_{\id} & & \\
            & \nuordf{i+1} \ar[r]^-{F\pco{M}{j}} \ar[dl]_{\nuconn{i}{i+1}} & \nuordf{j+1} \\
            \nuordf{i} \ar@/_{1pc}/[urr]_-{\nuconn{j+1}{i}} & & 
          }
        \end{displaymath}
commutes.  We conclude that $\pco{M}{i}$ is the identity. \qedhere
      \end{description}
    \end{proof}
\begin{prop} \label{prop:pressurj}
       Let $\kappa > \aleph_0$.  Then  $\psetkappa$ and $\psetpluskappa$ preserve, up to a surjection, the limit of any inverse $\omega$-chain $D$.  Explicitly: any channel $C$ through $D$ with all levels $\lkappa$-sized is the range of some $\lkappa$-sized set of full branches.
      \end{prop}
\begin{proof}
For $n \in \nats$, each $x \in C_{n}$ extends by dependent choice to a full
  branch through $D$.  For each $n \in \nats$ and $x \in C_n$, choose such a branch $\theta_n(x)$, by the Axiom of Choice.  Then the set 
  $\setbr{\theta_n(x) \mid n \in \nats, x \in C_{n}}$ is $\lkappa$-sized because $\lkappa$ is uncountable and regular, and has range $C$.
\end{proof}

% To apply this result, we say that an ordinal $i$ is \emph{$\omega$-cofinal}  when it is the supremum of a strictly increasing sequence $(i_n)_{n \in \nats}$ (hence a limit).  Any positive countable limit has this property.

% \begin{cor}
%   Let $i < \omega$.
%   \begin{itemize}
%   \item $\nuconn{i}{j}$ is surjective for all $j \geqslant i$
%   \item every element of $\nuordpk{i}$ is $\psetkappa$-coalgebraic.
%   \end{itemize}
% \end{cor}
% \begin{proof} The case $i=0$ is evident.   The successor case follows from the fact that $\psetkappa$ preserves surjections. \end{proof}

\begin{cor} \label{cor:countsurj}
  Let $\kappa > \aleph_0$.  
  \begin{enumerate}
  \item For $i < \omega_{1}$ and $h \geqslant i$, the connecting map $\nuconn{i}{h}\psetkappa$ is surjective.
  \item For $i < \omega_1$, every element of $\nuordpk{i}$ is $\psetkappa$-coalgebraic.
  \item \label{item:omegaonecauchy} Every element of $\nuordpk{\omega_1}$ is $\psetkappa$-Cauchy.
  \end{enumerate}
\end{cor}
\begin{proof}
We just need to prove  $\nuconn{i}{i+1}\colon \nuordpk{i} \to \nuordpk{i+1}$ surjective, for $i < \omega_1$.  The case $i=0$ is evident.  If $i$ is a positive limit, take a strictly increasing sequence $(i_{n})_{n \in \nats}$ with supremum $i$, then apply Proposition~\ref{prop:pressurj}.  The successor case follows from the fact that $\psetkappa$ preserves surjections. \end{proof}

\section{The Final Chain  as a Transition System} \label{sect:fctrans}
\subsection{Full powerset}

We shall now explain how the final chain of $\pset$ constitutes a transition system.  
\begin{defi}
  The \emph{predecessor} of an ordinal is defined as follows:
  \begin{itemize}
  \item $\pred{i+1} \eqdef i$. 
  \item If $i$ is a limit, $\pred{i} \eqdef i$. 
  \item Since we deem $0$ a limit, $\pred{0}=0$.
  \end{itemize}
In general, $\pred{i}$ is the
  supremum of all ordinals less than $i$.  
\end{defi}
Note, by the way, that $\mathsf{pred}$ is
  the left adjoint of the successor function.  We come to the main definition:
\begin{defi} \label{def:fctrans}
  The transition relation $\fctrans{i}$ from $\nuordp{i}$ to $\nuordp{\pred{i}}$ is defined as follows.
  \begin{itemize}
  \item For $a \in \nuordp{i+1}$ and $b \in \nuordp{i}$, we set $a \fctrans{i+1} b$ when $b \in a$.
  \item If $i$ is a limit, then for $a,b \in \nuordp{i}$, we set $a \fctrans{i} b$ when $b_{j} \in a_{j+1}$ for all $j < i$.
  \item Since we deem $0$ a limit, $\emptytup \fctrans{0} \emptytup$.
  \end{itemize}
In general, $a \fctrans{i} b$ when for all $j < i$ we have $\nuconn{j}{\pred{i}}b \in \nuconn{j+1}{i}a$.
\end{defi}
 Note that connecting maps and coalgebra projections preserve transition:
\begin{prop}\hfill
  \begin{enumerate}
  \item For $j \leqslant i$ and $a \in \nuordp{i}$, if $a \fctrans{i} b$ then $\nuconn{j}{i}a \fctrans{j} \nuconn{\pred{j}}{\pred{i}}b$.
  \item For a pointed system $x$ and ordinal $i$, if $x \rightsquigarrow y$ then $\pcop{i}{x} \fctrans{i} \pcop{\pred{i}}{x}$.
  \end{enumerate}
\end{prop}
Let us translate Definition~\ref{def:fctrans} into the language of Section~\ref{sect:channel}.  For general $F$, the map
\begin{displaymath}
\Phi_i \eqdef  (\nuconn{j}{\pred{i}})_{j < i} \colon \nuordf{\pred{i}} \to \lim_{j < i}\nuordf{j}  
\end{displaymath}
is an isomorphism (bijection)---indeed the identity if $i$ is a limit.  For $a \in \nuordp{i}$ and $b \in \nuordp{\pred{i}}$ we note that $\Psi_i(a)$ is a channel, and $\Phi_i(b)$ a full branch, through $(\nuordp{j})_{j < i}$.  We have $a \fctrans{i} b$ when $\Phi_i(b)$ is a full branch through $\Psi_i(a)$.  

\subsection{Restricted powerset} \label{sect:transrestr}

The final chain of $\psetkappa$ forms a subsystem of the final chain of $\pset$, in the following sense.
\begin{prop} \label{prop:kappasub}
Let $a \in \nuordpk{i}$ and $a \fctrans{i} b$ then 
    $b \in \nuordpk{\pred{i}}$.
\end{prop}
\begin{proof}
Since $\Psi_{i}(a)$ is a channel through $(\nuordpk{j})_{j < i}$ (with all levels $\lkappa$-sized), every full branch through it is a full branch through $(\nuordpk{j})_{j < i}$.
\end{proof}
As we shall see, the subsystem $(\nuordpk{i})_{i \in \On}$ is not $\lkappa$-branching (Proposition~\ref{prop:kappasucc} below).  Furthermore, the $\psetkappa$-coalgebraic elements need not form a subsystem (Proposition~\ref{prop:cosub} below).  However, the $\psetkappa$-Cauchy elements do form a subsystem.

% \begin{prop}
%   \begin{enumerate}
%   \item For a limit ordinal $i$, if $a \in \nuordpk{i}$ is $\psetkappa$-Cauchy, all its $\fctrans{i}$-successors are too.
%   \item For a limit ordinal $i < \kappa$, if $a \in \nuordpk{i}$ is $\psetkappa$-Cauchy, it is a $\fctrans{i}$-successor of a $\psetkappa$-coalgebraic element.
%   \end{enumerate}
% \end{prop}
% \begin{proof} \hfill
% \begin{enumerate}
% \item Let $a \fctrans{i} b$.  For $j<i$, the component $a_{j+1}$ is $\psetkappa$-coalgebraic, and $b_j \in a_{j+1}$, so $b_j$ is $\psetkappa$-coalgebraic.
% \item For each $j < i$, choose a pointed system $x_j$ such that
%   $a_j = \pcop{j}{x_j}$.  Let $y$ be a parent of $(x_j)_{j < i}$.  Then $\pcop{i}{y} \fctrans{i} a$ because for all $j < i$ we have $\pcop{i}{y}_{j+1} = \pcop{j+1}{y} = \setbr{\pcop{j}{z} \mid y \rightsquigarrow z} \ni \pcop{j}{x_j} = c_j$. \end{proof} 
% \end{enumerate}

% For an ordinal $i$ and $a \in \nuordp{i}$, we write $\emt{i}{a}$ for the corresponding pointed system.
Suppose now that $\kappa < \biginfty$.  Beyond $\kappa$, we recall that connecting maps are injections (Proposition~\ref{prop:injective}).  They are moreover transition system embeddings in the following sense.
\begin{prop} \label{prop:embedtrans}
 For $i \geqslant j \geqslant \kappa$ and $a \in \nuordpk{i}$, the map $\nuconn{\pred{j}}{\pred{i}}$ is a bijection from the $\rightsquigarrow_i$-successors of $a$ to the $\rightsquigarrow_j$-successors of $\nuconn{j}{i}a$. 
\end{prop}
\begin{proof}
This says that, if $C$ is a channel through $(\nuordpk{k})_{k < i}$ with all levels $\lkappa$-sized, then each full branch $b$ of $C \restriction_{j}$ extends uniquely to a full branch $a$ of $C$.  For $k < i$ we define $a_{k}$ as follows.  
\begin{itemize}
\item If $k \geqslant \kappa$, then for each  $l < \kappa$ the set $R_l$ of $a \in C_{k}$ such that $a_{l} = b_{l}$ is inhabited, so $\bigcap_{l \in \kappa} R_l$ has an element by Proposition~\ref{prop:regular}, unique since  $\nuconn{\kappa}{k}$ is injective.  Let $a_k$ be this element.
\item If $k < \kappa$, set $a_k \eqdef b_k$. \qedhere
\end{itemize}
\end{proof}
Hence, beyond $\kappa + \omega$, the transition system agrees with the final coalgebra structure given by Proposition~\ref{prop:ifstable}:
\begin{cor} \label{cor:transfinal}
  For $i \geqslant \kappa+ \omega$ and $a \in \nuordpk{i}$, the map $\nuconn{\pred{i}}{i}$ is a bijection from  the elements of $(\nuconn{i}{i+1})^{-1}a$ to the $\rightsquigarrow_i$-successors of $a$.
\end{cor}

\subsection{Examples}

The following result is adapted from~\cite{FortiHonsell:modselfdescript}.
\begin{prop} \label{prop:pcoppcop}
  Let $x$ and $y$ be pointed systems and $i \in \On$.  We have  $\pcop{i}{x} \fctrans{i} \pcop{\pred{i}}{y}$ iff, for every $j < i$, the pointed system $x$ has a successor $z$ such that $z\simi{j} y$. (Cf.~Proposition~\ref{prop:simkappasucc}.) 
\end{prop}
\begin{proof} Because $\pcop{i}{x} \fctrans{i} \pcop{\pred{i}}{y}$ iff for every $j < i$ we have $\pcop{j}{y} \in \pcop{j+1}{y}$, and the latter is $\setbr{\pcop{j}{z} \mid x \rightsquigarrow z}$. \end{proof}

In several cases we shall precisely describe the $\fctrans{i}$-successors of an element.
\begin{prop}
 For any $i>0$, the element $\zat{i}$ has no successors.
\end{prop}
\begin{proof} If $\zat{i} \fctrans{i} b$ then $b_{0} \in \zat{1} = \emptyset$. \end{proof}

\begin{prop} \label{prop:succpcopk}
  Let $\kappa < \biginfty$ and let $x$ be an $\lkappa$-branching pointed system.  For $i \geqslant \kappa$, the set of successors of $\pcop{i}{x}$ is $\setbr{\pcop{\pred{i}}{y} \mid x \rightsquigarrow y}$.
\end{prop}
\begin{proof}
Let $j$ be the maximum of $i$ and $\kappa+\omega$.  If $\pcop{i}{x} \rightsquigarrow b$, then, by Proposition~\ref{prop:embedtrans}, $\pcop{j+1}{x}$ has a unique element $c$ such that $\nuconn{\pred{i}}{j} c = b$.  By the final coalgebra property, $x$ has a successor $y$ such that $\pcop{j}{y} = c$ and hence $\pcop{\pred{i}}{y}=b$.
\end{proof}

\begin{prop}
  Let $x$ be a pointed system with a unique successor $y$.  For any ordinal $i$, the element $\pcop{i}{x}$ has unique successor $\pcop{\pred{i}}{y}$.
\end{prop}
\begin{proof} If $\pcop{i}{x} \rightsquigarrow b$ then for all $j < i$ we have $b_j \in \pcop{j+1}{x} = \setbr{\pcop{j}{y}}$ so $b_j = \pcop{j}{y}$.  Thus $b = \pcop{\pred{i}}{y}$.
 \end{proof}

We next consider the von Neumann ordinals.
\begin{prop} \label{prop:vonsucc}
  For ordinals $j \leqslant i$, the successors of
    $\pcop{i}{\von{j}}$ are listed without repetition
    \begin{itemize}
    \item as $(\pcop{\pred{i}}{\von{k}})_{k < j}$, if $j < i$
    \item as $(\pcop{\pred{i}}{\von{k}})_{k \leqslant \pred{i}}$, if
      $j=i$.
    \end{itemize}
\end{prop}
\begin{proof}
Proposition~\ref{prop:pcoppcop} tells us that these are indeed
  successors.  Uniqueness is by Proposition~\ref{prop:von}.

  Let $j \leqslant i$ and $\pcote{i}{\von{j}} \rightsquigarrow x$.  For
  every $k < i$, we have $x_k \in \pcote{k+1}{\von{j}}$, which is
  $\setbr{\pcote{k}{y} \mid \von{j} \rightsquigarrow y}$.  Thus $x_k$ is
  expressible as $\pcote{k}{y}$ for some successor $y$ of $\von{j}$,
  i.e\ $y = \von{l}$ for some $l < j$.  There is a unique such $l$ that is
  $\leqslant k$; we call it $g(k)$.

  Let $C$ be the set of $k < i$ such that $g(k) = k$ and $D$ its
  complement, i.e.\ the set of $k < i$ such that $g(k) < k$.

  For any $k \leqslant l < i$ we have
  $\pcote{k}{\von{g(k)}} = \pcote{k}{\von{g(l)}}$ i.e.\
  $t_{g(k)} \sim_{k} t_{g(l)}$.  This can be unpacked as follows.
  \begin{enumerate}
  \item \label{item:inc} If $k \in C$ then $g(l) \geqslant k$.
  \item \label{item:ind} if $k \in D$ then $g(l) = g(k)$.
  \end{enumerate}
  (\ref{item:ind}) implies that $g(l) = g(k) < k \leqslant l$, so
  $l \in D$. Thus $D$ is upper and $C$ is lower.

  Let $k$ be the supremum of $C$.  Since $\pred{i}$ is an upper bound
  for $C$ we have $k \leqslant \pred{i}$.

  We show that $x = \pcote{\pred{i}}{\von{k}}$, i.e.\ that for any
  $l < i$ we have $g(l) = l \sqcap k$.  If $l < k$ this holds because
  $l \in C$.  If $l \geqslant k$, we must show $g(l) = k$.  Firstly,
  (\ref{item:inc}) says that $g(l)$ is an upper bound for $C$, hence
  $g(l) \geqslant k$.  If $l=k$ we are done.  If $l > k$ then
  $k + 1 \leqslant l < i$, and $k \leqslant g(k + 1) \leqslant k + 1$.
  Since $k + 1 \not\in C$, we have $g(k + 1) = k$ and hence by
  (\ref{item:ind}), $g(l) = k$.

  Since $j$ is an upper bound for $C$, we have $k \leqslant j$.  If
  $j < i$, then $j \in D$ but $k \in C$ so $k < j$.
\end{proof}

\begin{prop} \label{prop:kappasucc}
 For $\kappa < \biginfty$, the set $\nuordpk{\kappa}$ has a $\psetkappa$-Cauchy 
element with precisely $\kappa$ successors.
\end{prop}
\begin{proof}
 Put $e \eqdef \pcote{\kappa}{\von{\kappa}}$.  For $i < \kappa$, the component $e_i$ is $\pcote{i}{\von{\kappa}} = \pcote{i}{\von{i}}$, hence $\psetkappa$-coalgebraic.   By Proposition~\ref{prop:vonsucc}, the successors of $e$ are listed without repetition as $(\pcote{\kappa}{\von{i}})_{i \leqslant \kappa}$, so $e$ has $\kappa$-many successors.
\end{proof}

Thus we may characterize which connecting maps are bijective.
\begin{prop}~\cite{AdamekLevyMiliusMossSousa:saturate} \label{prop:noearlier}
For ordinals $j \leqslant i$, the following are equivalent:
\begin{enumerate}
\item \label{item:noearlierp} The connecting map $\nuconn{j}{i}\psetkappa$ is bijective.
\item \label{item:noearlierc} Either  $j=i$ or $j \geqslant \kappa+\omega$ (for $\kappa < \biginfty$).
\end{enumerate}
\end{prop}
\begin{proof} For (\ref{item:noearlierc})$\Rightarrow$(\ref{item:noearlierp}), we saw in Section~\ref{sect:stable} that the final chain is stable at $\kappa + \omega$.  For the converse, assuming $j<i$, Proposition~\ref{prop:injective} covers the case where $j<\kappa$.  For $j = \kappa + n$, where $n \in \nats$, take $a \in \nuordpk{\kappa}$ with $\kappa$ successors.  Then $\setbr{-}^{n}a$ is not in the range of $\nuconn{\kappa+n}{\kappa+n+1}$, because for every $b \in \nuordpk{\kappa+n+1}$, each $n$-step descendant of $b$ has $\lkappa$-successors. \end{proof}

\section{Before $\omega_1$} \label{sect:count}

This section reprises from~\cite{AdamekLevyMiliusMossSousa:saturate} several properties that the final chain enjoys at countable ordinals.  (We have already seen one such property---surjectivity of connecting maps.)  We first mention some properties enjoyed at \emph{finite} ordinals:
\begin{prop} \label{prop:finpost}
  Let $n \in \nats$.
  \begin{enumerate}
  \item For $a,b \in \nuordp{n}$, if $a \simi{n} b$ then $a=b$.
  \item For any pointed system $x$, we can characterize $\pcop{n}{x}$ as the unique $b \in \nuordp{n}$ such that $b \simi{n} x$.
  \item \label{item:finpost} For any ordinal $i$ and $a \in \nuordp{i+n}$, we have $\pcop{n}{a} = \nuconn{n}{i+n}a$, and it can be characterized as the unique $b \in \nuordp{n}$ such that $b \simi{n} a$.
  \end{enumerate}
\end{prop}
\begin{proof} By induction on $n$. \end{proof}

The properties enjoyed at countable ordinals derive from the following fact.
\begin{prop} \label{prop:backconn} For $j < i < \omega_1$ and all $a \in \nuordp{i}$ we have 
  \begin{eqnarray*}
    \nuconn{j+1}{i} a & = &  \setbr{\nuconn{j}{i}b \mid a \fctrans{i} b}
  \end{eqnarray*}.
% or $\geqslant \kappa+\omega$.
 % For all $a \in \nuordpk{i}$ we have
 % \begin{eqnarray*}
 %      \nuconn{j+1}{i} a  & =  & \setbr{\nuconn{j}{i}b \mid a \fctrans{i} b}
 %    \end{eqnarray*}
\end{prop}
\begin{proof} 
Trivial if $i$ is a successor or $0$.  If $i$ is 
% \begin{itemize}
% \item The successor and $0$ cases are trivial. 
% \item The
%   $i \geqslant \kappa + \omega$ case holds because $\nuconn{j+1}{i}$
%   is the $j+1$-th coalgebra projection from the final coalgebra
%   $\nuordpk{i}$.  
%\item Let $i$ be 
a positive limit, take a strictly increasing sequence $(i_{n})_{n \in \nats}$ with supremum $i$.  Any $b \in \nuconn{j+1}{i}a$,  is an element of the $j$th level
  of the channel $\Psi_i(a)$ through $(\nuordpk{k})_{k < i}$, and it
  extends by dependent choice to a full branch.
%\end{itemize}
\end{proof} 

\begin{lem}\label{lem:pconuconn}
Let $j \leqslant i < \omega_1$.  For $a \in \nuordp{i}$ we have $\pcop{j}{a} = \nuconn{j}{i}a$.
\end{lem}
\begin{proof} By induction on $j$, using Proposition~\ref{prop:backconn} for the successor case. \end{proof}

\begin{prop}  \label{prop:countord}
  Let $i < \omega_1$.
  \begin{enumerate}
   \item \label{item:pcopself} For $a \in \nuordp{i}$ we have
    \begin{math}
      \pcop{i}{a}  =  a
    \end{math}.
  \item Every element of $\nuordp{i}$ is $\pset$-coalgebraic.
% \footnote{We cannot say every element is $\psetkappa$-coalgebraic; that is false in the case $\kappa = \aleph_0$ and $\omega \leqslant i < \omega+\omega$.  Note also that our proof of $\pset$-coalegbraicity (unlike our proof of $\psetkappa$-coalgebraicity, in the cases where it holds) explicitly constructs a pointed system.}.
  \item \label{item:simieq} For $a,b\in \nuordp{i}$, if $a \simi{i} b$ then $a=b$.
  \item  For $a,b\in \nuordp{i}$, if $a \sim b$ then $a=b$.
  \end{enumerate}
\end{prop}
\begin{proof} \hfill
\begin{enumerate}
\item By Lemma~\ref{lem:pconuconn}.   
% For $j \leqslant i$, we prove $\pcop{j}{a} = \nuconn{j}{i}a$ by induction on $j$, using Proposition~\ref{prop:backconn} for the successor case.
\item Follows from part~(\ref{item:pcopself}).
\item $a = \pcop{i}{a} = \pcop{i}{b} = b$.
\item Follows from part~(\ref{item:simieq}).
\qedhere
\end{enumerate}
\end{proof}

\begin{cor} \label{cor:countchar} \hfill
  \begin{enumerate}
  \item \label{item:countcharproj} Let $i < \omega_1$.  For any pointed system $x$, we may characterize $\pcop{i}{x}$ as the unique $b \in \nuordp{i}$ such that $b \simi{i} x$.
\item  \label{item:countcharconn} Let $j \leqslant i < \omega_1$.  For any $a \in \nuordp{i}$
    we may characterize $\nuconn{j}{i}a$ as the unique $b \in \nuordp{j}$ such
    that $b \simi{j} a$.
  \end{enumerate}
\end{cor}
\begin{proof} \hfill
\begin{enumerate}
\item  $\pcop{i}{\pcop{i}{x}} = \pcop{i}{x}$, by Proposition~\ref{prop:countord}(\ref{item:pcopself}), so $\pcop{i}{x} \simi{i} x$.  Uniqueness is by Proposition~\ref{prop:countord}(\ref{item:simieq}). 
\item $\pcop{j}{\nuconn{j}{i}a} = \nuconn{j}{i}a =
  \nuconn{j}{i}\pcop{i}{a} = \pcop{j}{a}$,
  so $\nuconn{j}{i}{a} \simi{j} a$.  Uniqueness is by
  Proposition~\ref{prop:countord}(\ref{item:simieq}).\qedhere
\end{enumerate}
\end{proof}

% We mention another situation in which the above property holds.
% \begin{prop}
% Let $j$ be finite and $i$ an ordinal $+j$.  For any  $a \in \nuordpk{i}$ the element $\nuconn{j}{i}a$ is the unique $b \in \nuordpk{j}$ such that $\emt{i}{a} \simi{j} \emt{j}{b}$. 
%   % Let $n \in \nats$ and $i \in \On$.  For any $a \in \nuordpk{i+n}$ the element $\nuconn{n}{i+n}a$ is the  unique $b \in \nuordpk{j}$ such that $\emt{i}{a} \simi{j} \emt{j}{b}$.
% \end{prop}

A pointed system $a$ is \emph{$\lkappa$-branching at depth $<i$} when every $b$ such that  $a \rightsquigarrow^{m} b$, for a natural number $m < i$, has $\lkappa$ successors.  If $i \geqslant \omega$, this just says that $a$ is $\lkappa$-branching.
\begin{prop} \label{prop:brdepth}
 Assume $\kappa < \biginfty$.   
  \begin{enumerate}
  \item Let $j$ be an ordinal.  Every element of $\nuordpk{\kappa+j}$ that is $\psetkappa$-coalgebraic is $\lkappa$-branching, and conversely if $\kappa =\aleph_0$.
  \item Let $j \leqslant i$ be ordinals.  Every element of $\nuordpk{\kappa+j}$ that is in the range of $\nuconn{\kappa+j}{\kappa+i}$ is $\lkappa$-branching at depth $<i$, and conversely if $\kappa = \aleph_0$.
  \end{enumerate}
\end{prop}
\begin{proof} \hfill
\begin{enumerate}
\item ($\Rightarrow$) follows from Proposition~\ref{prop:succpcopk}.  For ($\Leftarrow$), if $a \in \nuordpf{\omega+j}$ is finitely branching, then it is $\psetfin$-coalgebraic since $\pcop{\omega+j}{a} = a$.
\item ($\Rightarrow$) follows from Proposition~\ref{prop:embedtrans}, since every element of $\nuordpk{\kappa+i}$ is $\lkappa$-branching at depth $<i$.  For ($\Leftarrow$), let $a \in \nuordpk{\kappa+j}$ be $\lkappa$-branching at depth $<i$.  Since  $b \eqdef \pcop{\omega+i}{a}$ is sent by $\nuconn{\omega+j}{\omega+i}$ to $a$, it suffices to show that $b \in \nuordpf{\omega+i}$.  If $i \geqslant \omega$ this holds because $a$ is finitely branching; we prove the case $i < \omega$ by induction on $i$. For $i = 0$ we have $j=i$ and $b=a \in \nuordpf{\kappa}$.  For $i = i'+1$ we have $b = \setbr{\pcop{\omega+i'}{c} \mid a \rightsquigarrow c}$, which is in $\nuordpf{\omega+i}$ because $a$ has finitely many successors and each of them is finitely branching at depth $<i'$.\qedhere
\end{enumerate}
\end{proof}

\section{A Ghost at $\omega_{1}$} \label{sect:ghostomegaone}

The previous section has established properties of $\nuordpk{i}$ for $i < \omega_1$.  But for $i = \omega_1$, many questions remain unresolved.  
\begin{itemize}
\item For $i < \omega_1$, every element of $\nuordpk{i}$ is $\psetkappa$-coalgebraic (for $\kappa > \aleph_0$).  Is every element of $\nuordpc{\omega_1}$ at least $\pset$-coalgebraic?
\item For $i < \omega_1$, the set $\nuordp{i}$ is strongly extensional, i.e.\ bisimilar elements are equal.  What about $\nuordpc{\omega_1}$?
\item For $j < i < \omega_1$ we saw (Proposition~\ref{prop:backconn}) that $\nuconn{j+1}{i}\pset$ sends $a$ to $\setbr{\nuconn{j}{\pred{i}}b \mid a \rightsquigarrow b}$.  Is this true for $\nuconn{j+1}{\omega_1}\psetcount$, where $j < \omega_1$?
\item The finitely branching elements of $\nuordpf{\omega}$ are the $\psetfin$-coalgebraic ones.  Are the countably branching elements of $\nuordpc{\omega_1}$ all $\psetcount$-coalgebraic?
% \item For $\kappa > \aleph_1$, do the $\psetkappa$-coalgebraic elements form a subsystem of $\nuordpk{\omega_1}$?
\end{itemize}
The following result will provide a negative answer to all the above questions.
\begin{prop} \label{prop:ghostomegaone}
  The set $\nuordpc{\omega_1}$ has an element $a$, distinct from $0(\omega_1)$, that has no successor.
\end{prop}
Such an element $a$ is called a \emph{ghost} because it ``appears to be alive'' (for all $j<\omega_1$, the component $a_{j+1}$ is inhabited) yet is ``dead'' (has no successors).  The negative answers are deduced as follows. 
\begin{itemize}
\item For a pointed system $x$, if $a = \pcop{\omega_1}{x}$, then $\emptytup \in a_1 = \pcop{1}{x} = \setbr{\pcop{0}{y} \mid x \rightsquigarrow y}$.  So $x$ has a successor $y$, so $a = \pcop{\omega_1}{x} \rightsquigarrow \pcop{\omega_1}{y}$, contradiction.
\item $a$ and $\zat{\omega_1}$ are bisimilar (having no successors) yet distinct.  So $\nuordpc{\omega_1}$ is not strongly extensional.
\item For $j < \omega_1$, the set $\nuconn{j+1}{i}a = a_{j+1}$  has an element $b$, but $b$ is not of the form $\nuconn{j}{i}c$ for a successor $c$ of $a$.
\item The element $a$ is countably branching but not $\psetcount$-coalgebraic.
\end{itemize}
A ghost corresponds via $\Psi_{\omega_1}$ to a channel through $(\nuordpc{i})_{i < \omega_1}$, with all levels countable, that has a root but no full branch.  The rest of the section is devoted to proving that such a channel exists.
% We may restate Proposition~\ref{prop:ghostomegaone} in terms of channels.  A channel through an $\omega_1$-chain is \emph{Aronszajn} when it has all levels countable and a root but no full branch.
% \begin{prop} \label{prop:aronchann}
%   There is an Aronszajn channel through $(\nuordpc{i})_{i < \omega_1}$.
% \end{prop}
% The rest of the section is devoted to proving Proposition~\ref{prop:aronchann}.  
Our proof involves three steps:
\begin{enumerate}
\item obtaining an ``Aronszajn tree'' 
\item  embedding it into the complete binary tree
\item embedding the complete binary tree into the final chain.
\end{enumerate}
We begin with the following notions.
\begin{defi}
Let $I$ be a well-ordered set with least element.  A \emph{tidy $I$-tree}\footnote{The   set-theoretic literature commonly uses a more general notion of tree.  See Section~\ref{sect:treestidy} for a comparison.} is an inverse $I$-chain $D$ with the following properties.
\begin{itemize}
\item For any $j\leqslant i \in I$, the connecting map  $\ddiag{j}{i}$ is surjective, i.e.\ every  $a \in D_j$ has an
      $i$-development. 
\item For any limit $i \in I$, every full branch  through $(D_j)_{j < i}$
    has at most one extension to a full
    branch through $(D_j)_{j \leqslant i}$.
\item $D$ has a (necessarily unique) root.
\end{itemize}
\end{defi}
Note that a tidy $I$-tree is a channel through itself.  

It is important 
 to know whether a tidy tree is guaranteed to have a full branch.  The following two results~\cite{HigmanStone:arontree,Kurepa:arontree} show that this is not always the case.  Our presentation follows~\cite{Bergman:Aronszajnnote}.  
\begin{prop}\label{prop:nofull}
For regular $\lambda > \aleph_0$, there is a tidy $\lambda$-tree that has all levels $\leqlambda$-sized and has no full branch.
\end{prop}\begin{proof}
For each $i < \lambda$, let $D_i$ be the set of strictly increasing sequences $0=x_0 < \cdots < x_n < \lambda$ of ordinals, where $x_n \geqslant i$ but $x_m < i$ for all $m < n$.  Clearly $D_i$ is $\leqlambda$-sized.   For $j \leqslant i < \lambda$, let $\ddiag{j}{i} : D_i \to D_j$ send $x \in D_i$ to its unique prefix in $D_j$.  Then $(D_i)_{i < \lambda}$ is a tidy $\lambda$-tree, and any full branch would give a cofinal $\omega$-sequence in $\lambda$, contradicting regularity.
\end{proof}

\begin{thm} \label{thm:tidyexists}
There is a tidy $\omega_1$-tree that has all levels countable and no full branch.  (Such a tree is said to be \emph{Aronszajn}.)
\end{thm}
\begin{proof} For $i < \omega_1$, let $E_i$ be the set of strictly increasing
  $i$-sequences $x = (x_j)_{j < i}$ of nonnegative rationals.  The supremum of $x \in E_i$---taken to be $0$ if $x$ is empty---is a nonnegative real or $\infty$, and written $\supof{x}$.

We shall define a countable subset $D_i \subseteq E_i$ for all $i < \omega_1$, with the following properties.
\begin{enumerate}
\item \label{item:finsup} If $x \in D_i$ then its supremum is (finite and) rational.
\item If $j \leqslant i < \omega_1$ and $x \in D_i$ then the $j$-sequence prefix of $x$ is in $D_j$.
\item \label{item:allextend} If $j < i < \omega_1$, then for any $x \in D_j$  and rational $r > \supof{x}$, $x$ has an extension  in $D_i$ with supremum $r$.
\end{enumerate}
Suppose that we have $D_j$ for all $j < i$, with these properties.  We define $D_i$ as follows; properties (\ref{item:finsup})--(\ref{item:allextend}) are easily verified in each case.
\begin{itemize}
\item For $i = 0$, let $D_i$ consist of the empty sequence.
\item Suppose $i = j+1$.  Let $D_i$ consist of all the extended sequences $xr$ for $x \in D_j$ and rational $r > \supof{x}$.
\item Suppose $i > 0$ is a limit.  By the Axiom of Choice, choose, for each $j < i$ and $x \in D_j$ and rational $r > \supof{x}$, a strictly increasing sequence of ordinals $(k_{n})_{n \in \nats}$, where $k_0 = j$, with supremum $i$; and  a strictly increasing sequence of rationals $(s_{n})_{n \in \nats}$, where $s_{0} = \supof{x}$, with supremum $r$.    (For example, take $s_n = 2^{-n}\supof{x} + (1-2^{-n})r$.)  We define $P_n(j,x,r) \in D_{k_{n}}$ with supremum $s_n$ by induction on $n \in \nats$.
  \begin{itemize}
  \item Let $x_0$ be $x$.
  \item Let $x_{n+1}$ be an extension of $x_n$ in $D_{k_{n}}$ with supremum $s_{n+1}$.
  \end{itemize}
Then the $i$-sequence $\bigsqcup_{n \in \nats} P_{n} (j,x,r)$ has supremum $r$ and extends $x$.  Let $D_i$ be the set of all these.  
\end{itemize}
We take the inverse chain $(D_i)_{i < \omega_1}$, with connecting map $\ddiag{j}{i}$ sending $x \in D_i$ to its $j$-sequence prefix.   The limit property is clear.  The connecting map $\ddiag{j}{i}$ is surjective because, in the nontrivial case $j < i$,  every  $x \in D_j$ has an extension in $D_i$ with supremum $\supof{x} +1$.  The levels are countable by construction.  Finally, any full branch of $(D_i)_{i < \omega_1}$ would give a  strictly increasing $\omega_1$-sequence of nonnegative rationals, which does not exist.  
\end{proof}

This completes our first step.  Before describing the second, we must give a suitable notion of embedding.
\begin{defi}
 A \emph{cofinal embedding} of inverse chains $\alpha \colon (D_i)_{i \in I} \to (E_{j})_{j \in J}$ consists
    of the following.
    \begin{itemize}
    \item A monotone map $\alzero \ccons I \to J$ (the \emph{index
        map}) that is cofinal, i.e.\ for any $j \in J$ there is
      $i \in I$ such that $\alzero i \geqslant j$.
    \item For each $i \in I$, an injection
      $\alone{i} \ccons D_i \to E_{\alzero i}$ (the \emph{$i$-th level
        map}) that is natural in $i$, i.e.\ for $j \leqslant i$ in
      $I$, if $y \in D_i$ is an $i$-development of $x \in D_j$ then
      $\alone{i}y$ must be an $\alzero{i}$-development of
      $\alone{j}x$.
    \end{itemize}
  % \item For any full branch $x$ through $D$, its \emph{$\alpha$-image} is the unique full branch $y$ through $E$ such that, for all $i \in I$, we have  $y_{\alzero{i}} = \alone{i}x_i$.
For any channel $B$ through $D$, its \emph{$\alpha$-image} is the unique channel $C$ through $E$ such that, for all $i \in I$, the set $C_{\alzero{i}}$ is  $\setbr{\alone{i}x \mid x \in B_i}$. 
\end{defi}

% \begin{defi}
%   Let $\alpha$ be a cofinal embedding of  $(D_i)_{i \in I}$ into $(E_j)_{j \in J}$.
%   \begin{enumerate}
%   \item For any full branch $x$ through $D$, we write $\albranch{x}$ for its \emph{image}: the unique full branch $y$ through $E$ such that, for all $i \in I$, we have  $y_{\alzero{i}} = \alone{i}x_i$.
%   \item For any channel $B$ through $D$, we write $\alchann{B}$ for its \emph{image}: the unique channel $C$ through $E$ such that, for all $i \in I$, the set $C_{\alzero{i}}$ is the range of $\alone{i}$ on $B_i$.
%   \end{enumerate}
% \end{defi}

\begin{prop} \label{prop:cofinalemb}
Let $\alpha \colon (D_i)_{i \in I} \to (E_j)_{j \in J}$ be a cofinal embedding of inverse chains.  Let $B$ be a channel through $D$, with $\alpha$-image $C$.
  \begin{enumerate}
  \item (Assuming $I$ has a least element.) If $B$ has a root, then $C$ does too.
  \item If $B$ has all levels $\lkappa$-sized, then so does $C$.
  \item \label{item:bijfb} The map $(y_j)_{j \in J} \mapsto (x_{\alone{i}})_{i \in I}$ is a bijection from the full branches of $C$ to those of $B$.
  \end{enumerate}
\end{prop}
By part (\ref{item:bijfb}), if $B$ has no full branch, then neither does $C$.  Our second step involves the following tidy tree.
\begin{defi}\hfill
  \begin{enumerate}
   \item Let $\bits^{i}$ be the set of $i$-sequences $(d_j)_{j < i}$ of bits.
  \item For $j \leqslant i$ and $x \in \bits^{i}$ we write $x\restriction_j$ for the restriction of $x$ to $j$.
  \item For a positive limit $i$, the \emph{complete binary $i$-tree} is $(\bits^{j})_{j < i}$, with connecting maps given by restriction.
  \end{enumerate}
\end{defi}

\begin{prop}\label{prop:allintocompbin}
Let $D$ be a tidy $\omega_1$-tree with all levels countable.  Then there is a cofinal embedding $\alpha$ of $D$ into the complete binary $\omega_1$-tree.
\end{prop}
\begin{proof}
The index map is $i \mapsto \omega \times i$, which is well-defined and cofinal because $i<\omega_1$ implies $i \leqslant \omega \times i < \omega_1$.  The injection $\alone{i} \ccons D_i \to \bits^{\omega \times i}$ is defined by induction on $i$, ensuring naturality wrt all $j \leqslant i$.  For the successor case, we choose an injection from the countably many $i+1$-developments of each $x \in D_i$ to the $2^{\aleph_0}$-many $\omega\times (i+1)$ developments of $\alone{i}x \in \bits^{\omega \times i}$.  The case where $i$ is a limit is uniquely defined: each $x \in D_i$ is mapped to the unique $i$-sequence that, for all $j < i$, extends $\alone{j}\ddiag{j}{i}x$.  This is injective because $x$ is determined by $(\ddiag{j}{i})_{j < i}$.
\end{proof}

\begin{cor}\label{cor:aronbin}
  There is a channel through the complete binary $\omega_1$-tree, with all levels countable, that has a root but no full branch.
\end{cor}
\begin{proof} Proposition~\ref{thm:tidyexists} gives an Aronszajn tidy tree.  Embed it into the complete binary $\omega_1$-tree and take its image. \end{proof}

It remains to cofinally embed the complete binary $\omega_1$-tree into $(\nuordpc{i})_{i < \omega_1}$.  For this we encode a sequence of bits as a set of ordinals, following~\cite{FortiHonsell:modselfdescript}[Lemmas~2.2 and~2.3].
\begin{defi}
  For any ordinal $i$, and $c \in \bits^{i}$ we define 
  \begin{displaymath}
    \bitsen{c} = \setbr{j < i \mid c_j = 1} \cup \setbr{i}
  \end{displaymath}
\end{defi}
% Recall $(\von{i})_{i \in \On}$ from Definition~\ref{def:von}.
% \begin{defi} \label{def:vonset}
%   For each  set of ordinals $R$, let $\vonoset{R}$ be the pointed system whose successors are $(\vono{i})_{i \in R}$.  
% \end{defi}

 For each  set of ordinals $R$, let $\vonoset{R}$ be a parent of $(\vono{i})_{i \in R}$.
\begin{lem}\label{prop:restrict}
  Let $j\leqslant i$ and $c\in \bits^{j}$ and $d \in \bits^{i}$.  Then $c = d \restriction_j$ iff $\vonset{\bitsen{c}} \simi{j+1} \vonset{\bitsen{d}}$.
\end{lem}

\begin{cor} \label{cor:bintreeinto}
There is a cofinal embedding $\beta$ from the complete binary $\omega_1$-tree into $(\nuordpc{i})_{i < \omega_1}$, with index map $i \mapsto i+1$.
\end{cor}
\begin{proof}
The injection $\beta_j \colon \bits^{j} \to \nuordpk{j+1}$ sends $c$ to $\pcote{j+1}{\vonset{\bitsen{c}}}$.  Injectivity and naturality follow from Lemma~\ref{prop:restrict}.
\end{proof}
We now complete our proof.  Corollary~\ref{cor:aronbin} gives a channel $E$ through the complete binary $\omega_1$-tree, with all levels countable, that has a root but no full branch.  Its $\beta$-image is a channel through $(\nuordpc{i})_{i < \omega_1}$ with the same properties.

\section{The Number of Successors} \label{sect:numbersucc}

We further consider the final chain of $\psetcount$.   Proposition~\ref{prop:kappasucc} tells us that some element of $\nuordpc{\omega}$ has $\aleph_0$ successors (since $\nuordpc{\omega} = \nuordpf{\omega}$), and some element of $\nuordpc{\omega_1}$ has $\aleph_1$ successors.  Are there elements with more successors than this? Although this question is a digression from our main narrative, it provides an application of the infrastructure we have assembled.
\begin{prop}
  The size of $\nuordpc{i}$ is 
  \begin{itemize}
  \item finite and positive, if $i$ is finite
  \item $2^{\aleph_0}$, if $\omega \leqslant i < \omega_1$
  \item $2^{\aleph_1}$, if $i \geqslant \omega_1$. 
  \end{itemize}
\end{prop}
\begin{proof} 
Induction on $i$ gives the upper bound, since $\psetcount$ and countable limit preserves $\leqslant\! 2^{\aleph_0}$-sizedness, and $\omega_1$-limit preserves $\leqslant\! 2^{\aleph_1}$-sizedness.  From $\omega_1$ onwards the sets cannot get bigger.  The lower bound is proved as follows.  For a set of ordinals $R$, we write $\vonoset{R}$ for a parent of $(\vono{i})_{i \in R}$.
\begin{itemize}
\item For $i \geqslant \omega$ we have a family  $(\pcote{i}{\vonset{R}})_{R \subseteq \omega}$ of $2^{\aleph_0}$ distinct elements of $\psetcount{i}$. For distinctness, let $R,S \subseteq \omega$ with $n \in R \setminus S$.  Then $\vonset{R} \not \simi{n+2} \vonset{S}$ because there is no $m \in S$ such that $\von{n} \simi{n+1} \von{m}$. 
\item For $i \geqslant \omega_{1}$ we have a family  $(\pcote{i}{\vonset{J}})_{J \subseteq \omega_1}$ of $2^{\aleph_1}$ distinct elements of $\psetcount{i}$, by the same argument.
  \qedhere
\end{itemize}
\end{proof}

\begin{prop}
For every positive limit $i < \omega_1$, some element of  $\nuordpc{i}$ has $2^{\aleph_0}$ successors.
%  \begin{enumerate
%  \item 
% \item Every element of  $\nuordpc{i}$ has either countably many or $2^{\aleph_0}$ successors.
%  \end{enumerate}
\end{prop}
\begin{proof}
Let $(i_n)_{n \in \nats}$ be a strictly increasing sequence with supremum $i$.  
% \begin{enumerate}\item 
 We define a cofinal embedding $\alpha$ of the complete binary $\omega$-tree into the complete binary
  $i$-tree with index map $n \mapsto i_n$ as follows: the injection $\alpha_{n}$ sends $(c_m)_{m < n}$ to $(d_j)_{j < i_n}$, where $d_j$ is $c_m$ if $j = i_m$ for some (unique) $m < n$ and $0$ otherwise.  Corollary~\ref{cor:bintreeinto} gives a cofinal embedding $\beta$ of the complete binary $i$-tree into $(\nuordpc{j})_{j < i}$.  The $\beta\alpha$-image of the complete binary $\omega$-tree is a channel through $(\nuordpc{j})_{j < i}$
  with $2^{\aleph_0}$ full branches, corresponding across $\Psi_i$ to an
  element of $\nuordpc{i}$ with $2^{\aleph_0}$ successors.
% \item Let $x \in \nuordpc{i}$.. %  Then $\Psi x$ is a developable subtree of  $(\nuordpc{j})_{j < i}$ with all levels countable and $\lambda$ full branches.  So 
% Then $((\Psi x)_{i_n})_{n \in \nats}$ is a developable subtree of    $(\nuordpc{i_n})_{n \in \nats}$ with all levels countable, and its full branches correspond to the successors of $x$. But any $\omega$-tree $D$ with all levels countable has either countably many or $2^{\aleph_0}$ full branches.  This fact follows from the Cantor-Bendixson theorem; we shall prove it directly. 
% For any $n \in \nats$ and $x \in D_n$, we say $(i,x)$ is \emph{hereditarily splitting} if there is  
\end{proof}

At $\omega_1$ the question is harder to answer.
\begin{prop}
For any cardinal $\lambda$, the following are equivalent.
\begin{enumerate}
  \item \label{item:kurel} Some element of $\nuordpc{\omega_1}$ has precisely $\lambda$ successors.
 \item \label{item:kurcount} There is a tidy $\omega_1$-tree, with all levels countable, that has precisely $\lambda$ full branches.
%   \item  \label{item:kurdev} Some developable subtree of the complete binary $\omega_1$-tree, with every level countable,  has precisely $\lambda$ full branches.
%   \item \label{item:kurtree} Some subtree of the complete binary $\omega_1$-tree, with every level countable,  has precisely $\lambda$ full branches.
  \end{enumerate}
\end{prop}
\begin{proof} 

For  (\ref{item:kurel})$\Rightarrow$(\ref{item:kurcount}), any such element corresponds across $\Psi_{\omega_1}$ to the desired tree---unless it is $\zat{\omega_1}$, but in that case $\lambda = 0$ and Theorem~\ref{thm:tidyexists} gives~(\ref{item:kurcount}).

For  (\ref{item:kurcount})$\Rightarrow$(\ref{item:kurel}), let $D$ be such a tree.   We embed it cofinally into the complete binary $\omega_1$-tree (Proposition~\ref{prop:allintocompbin}), and embed this in turn into $(\nuordpc{i})_{i < \omega_1}$ (Corollary~\ref{cor:bintreeinto}), and the image corresponds via $\Psi_{\omega_1}$ to the desired element. \end{proof}
Thus the existence of an element of $\nuordpc{\omega_1}$ with more than $\aleph_1$ successors is equivalent to the existence of a tidy $\omega_1$-tree, with all levels countable, that has more than $\aleph_1$ full branches.  Such a tree is said to be \emph{Kurepa} and its existence is independent of ZFC, under certain assumptions\footnote{ZFC cannot prove that a Kurepa tree exists, if ZFC $+$ an inaccessible is consistent~\cite{Silver:kurepa}.
  ZFC cannot prove that a Kurepa tree does not exist, if ZFC is consistent (Solovay).\label{foot:kurepa}}.

\section{Beyond $\omega_1$} \label{sect:beyond}
\subsection{Properties of Ordinals and Cardinals}

We have now completed our study of the final chain of the countable powerset functor.  But for $\psetkappa$ we can continue, by asking what happens in the final chain beyond $\omega_1$.  We shall see that, surprisingly, there is a sharp division between those limit ordinals where strong extensionality holds and those where there is a ghost.   This is our main result, Theorem~\ref{prop:finalcover} below.  The same method of~\cite{FortiHonsell:modselfdescript} is adapted to construct all these ghosts. 

Our classification relies on the following properties.
\begin{defi}
  An  ordinal is \emph{$\omega$-cofinal} when it is the supremum of 
a strictly increasing sequence $(i_n)_{n \in \nats}$.
\end{defi}
Thus every countable limit is either $0$ or $\omega$-cofinal.
\begin{defi} \label{def:treepropinacc}
  Let $\lambda$ be a regular infinite cardinal.
  \begin{enumerate}
  \item \label{item:treeprop} $\lambda$ has the \emph{tree property} when every tidy $\lambda$-tree with all levels $\llambda$-sized has a full branch.  % A counterexample to this property---i.e.\ a tidy $\lambda$-tree with all levels $\llambda$-sized that has no full branch---is called a \emph{tidy $\lambda$-Aronszajn tree}.
  \item $\lambda$ is  \emph{strongly inaccessible} when, for every $\llambda$-sized set $A$, the set $\pset A$ is also $\llambda$-sized.
  \item $\lambda$ is \emph{weakly compact} when it is strongly inaccessible and has the tree property.
  \end{enumerate}
\end{defi}
Note that in Definition~\ref{def:treepropinacc}(\ref{item:treeprop}), it is essential for the levels to be $\llambda$-sized, because of Proposition~\ref{prop:nofull}.   Only in the case $\lambda = \aleph_0$ is a full branch guaranteed to exist without this condition.  By Theorem~\ref{thm:tidyexists}, $\aleph_1$ does not have the tree property.  Whether $\aleph_2$ has it is independent of ZFC under certain assumptions\footnote{\label{foot:treeprop}ZFC cannot prove that $\aleph_2$ has the tree property, if ZFC is consistent~\cite{Specker:problemsikorski}.  ZFC cannot prove that $\aleph_2$ lacks the tree property, if ZFC $+$ a weakly compact cardinal is consistent~\cite{Mitchell:aronszajn}.}.

We deem $\aleph_0$ to be strongly inaccessible (and hence weakly compact).  Successor cardinals such as $\aleph_1$ are not inaccessible.  The importance of strong inaccessibility  in our story comes from the following fact.
\begin{prop}\label{prop:stronginacc}
  If $\lambda$ is strongly  inaccessible, then for all $i< \lambda$ the set $\nuordp{i}$ is $\llambda$-sized.
\end{prop}
\begin{proof} By induction on $i$. The successor case is by the definition of strong inaccessibility.   For a limit $i < \lambda$,  we use the fact that if $(A_j)_{j < i}$ is an inverse chain of $\llambda$-sized sets then its limit is $\llambda$-sized. \end{proof}

\subsection{Connecting map to a successor}

In Section~\ref{sect:count} we derived several properties of countable ordinals from Proposition~\ref{prop:backconn}.  Accordingly, we shall begin by analyzing at which ordinals the property described there holds.
\begin{defi}
  The following limit ordinals are said to be \emph{$\kappa$-extensible}:
  \begin{itemize}
  \item $0$.
  \item Any $\omega$-cofinal ordinal.
  \item Any weakly compact cardinal.
  \item Any limit ordinal $\geqslant \kappa+\omega$, if $\kappa < \biginfty$.
  \item $\kappa$, if $\kappa < \biginfty$ and $\kappa$ has the tree property.
  \end{itemize}
\end{defi}
We complete Proposition~\ref{prop:backconn} as follows.
\begin{prop} \label{prop:backconnallord}
  For ordinals $j < i$ the following are equivalent.
  \begin{enumerate} 
  \item \label{item:backconnallordp} For  all $a \in \nuordpk{i}$ we have 
  \begin{math}%\label{eqn:backconn}
    \nuconn{j+1}{i} a  =   \setbr{\nuconn{j}{i}b \mid a \fctrans{i} b}
  \end{math}.
  \item \label{item:backconnallordc} $i$ is either a successor or a $\kappa$-extensible limit.
\end{enumerate}
\end{prop}
\begin{proof} of (\ref{item:backconnallordc}) $\Rightarrow$ (\ref{item:backconnallordp}).  The successor and zero cases are trivial, and the case $i \geqslant \kappa + \omega$ follows from Corollary~\ref{cor:transfinal}.  Let us consider the other cases. For $b \in \nuconn{j+1}{i}$, we know that $b$ is an element of the $j$th level of the channel $\Psi_i(a)$ through $(\nuordpk{k})_{k < i}$, and we want to extend it to a full branch.  We may reformulate this problem by expressing $i$ as $j+i'$, and defining the tidy $i'$-tree $(A_k)_{k < i'}$ where $A_k$ is the set of $j+k$-developments of $b$.  We want a full branch for this tree.
\begin{itemize}
% \item If $i = i'+1$, any $i'$-development of $b$ gives us such a full
%   branch. 
% \item The case $i=0$ is impossible. 
\item Suppose that $i$ is $\omega$-cofinal, so $i'$ is too.  Let  $(i_n)_{n \in \nats}$ be a strictly increasing sequence with supremum $i'$.  We obtain a full branch by dependent choice. (The size of the sets $A_k$ is immaterial in this case.)
\item  Suppose that $\kappa$ has the tree property and $i = \kappa$, so $i' = \kappa$.  For $k < \kappa$ the set $A_k$, being a subset of  $\nuordpk{j+k}$, is $\lkappa$-sized. So the tree property gives a full branch through $(A_k)_{k < \kappa}$.
\item Suppose that $i$, and therefore $i'$, is a weakly compact cardinal $\lambda < \kappa$.  For $k < \lambda$, we have also $j+k < \lambda$, and so $A_k$, being a subset of $\nuordpk{j+k}$, is $\llambda$-sized by Proposition~\ref{prop:stronginacc}.  So the tree property gives a full branch through $(A_k)_{k < \lambda}$.  \qedhere
\end{itemize}
\end{proof}
To prove the converse, we use the following notion of ghost.
\begin{defi}
Let $i$ be a limit. A \emph{$\psetkappa$-ghost at $i$} is a $\psetkappa$-Cauchy element of $\nuordpk{i}$, distinct from $0(i)$, that has no successor.
\end{defi}
The Cauchy condition was omitted in Proposition~\ref{prop:ghostomegaone} because every element of $\nuordpk{\omega_1}$ is $\psetkappa$-Cauchy (Corollary~\ref{cor:countsurj}(\ref{item:omegaonecauchy}).)

Any $\psetkappa$-ghost $a$ at $i$ does not satisfy the equation in (\ref{item:backconnallordp}), because the LHS is $a_{j+1}$, which is inhabited, but the RHS is empty.  So there cannot be a $\psetkappa$-ghost at a $\kappa$-extensible ordinals.  The following two results establish that, at all other limits, a $\psetkappa$-ghost does exist, giving (\ref{item:backconnallordp}) $\Rightarrow$ (\ref{item:backconnallordc}) as required.

Again we use channels: a $\psetkappa$-ghost at $i$ corresponds via $\Psi_i$ to a channel through $(\nucoalgpk{j})_{j < i}$, with all levels $\lkappa$-sized, that has a root but no full branch. 
\begin{prop} \label{prop:ghostnotree}
  Let $\kappa < \biginfty$ not have the tree property.  Then there is a $\psetkappa$-ghost at $\kappa$.
\end{prop}
\begin{proof} 
Since $\kappa$ does not have the tree property, there is a tidy $\kappa$-tree $D$ with no full branch.  (Such a tree is said to be \emph{$\kappa$-Aronszajn}.) 

We give a cofinal embedding $\alpha$ of $D$ into the complete binary $\kappa$-tree as follows.  The index map is $\alzero \ccons i \mapsto \sum_{j < i} |D_j|$.  This is well-defined and cofinal because $i<\kappa$ implies $i \leqslant \alzero{i} < \kappa$, since $1 \leqslant |D_j| < \kappa$ and $\kappa$ is regular.    We give the $i$-th level map $\alone{i} \ccons D_i \to \bits^{\alzero{i}}$ by induction on $i$, ensuring naturality wrt all $j \leqslant i$.   For the successor case, we have $\alzero{(i+1)} = \alzero{i} + |D_i|$, so we take an injection from the $\leqslant\! |D_i|$-many $i+1$-developments of each $x \in D_i$ to the $2^{|D_i|}$-many $\alzero{(i+1)}$-developments of $\alone{i}x \in \bits^{\alzero{i}}$.   The case where $i$ is a limit is uniquely defined, as in the proof of Proposition~\ref{prop:allintocompbin}.

The image $E$ of $D$ is a channel through the complete binary $\kappa$-tree with a root but no full branch. As in Corollary~\ref{cor:bintreeinto}, we define a cofinal embedding $\beta$ of the complete binary $\kappa$-tree into $(\nucoalgpk{j})_{j < \kappa}$, with index map $j \mapsto j+1$.  The $\beta$-image of $E$ is the desired channel.
\end{proof}

\begin{prop} \label{prop:ghostbeforei}
  Let $i<\kappa$ be a positive limit that is neither $\omega$-cofinal nor weakly compact.  Then there is a $\psetkappa$-ghost at $i$.
\end{prop}
\begin{proof} 
Let $\lambda$ be the cofinality of $i$.   Then $\lambda$ is a regular cardinal $> \aleph_0$.  There are three possibilities:
\begin{enumerate}
\item \label{item:notreg}  $i$ is not a regular cardinal, i.e.\ $i > \lambda$.
\item \label{item:notinacc} $i=\lambda$ and $\lambda$ is not strongly inaccessible.
\item \label{item:nottree} $i=\lambda$ and $\lambda$ does not have the tree property.
\end{enumerate}
For case (\ref{item:notreg}), Proposition~\ref{prop:nofull} gives a tidy $\lambda$-tree $D$, with every level $\leqlambda$-sized, that has no full branch.  Evidently there is a cofinal embedding $\alpha$ of $D$ into the complete $\lambda$-ary $\lambda$-tree  $(\lambda^{j})_{j < \lambda}$, with index map $j \mapsto j$.  We shall give a cofinal embedding $\beta$ of the complete $\lambda$-ary $\lambda$-tree into $(\nucoalgpk{j})_{j < i}$.  Then the $\beta\alpha$-image of $D$ is a channel through $(\nucoalgpk{j})_{j < i}$, with every level $\leqlambda$-sized, that has a root but no full branch.  Whereas we previously encoded  a sequence of bits as a set of ordinals, we shall now encode a sequence of elements of $\lambda$ as a set of sets of ordinals.

Since $i > \lambda$, express $i$ as the  supremum of a strictly increasing $\lambda$-sequence $(\eta_{j})_{j < \lambda}$ in $\eta$, with $\eta_0 \geqslant \lambda$.  The idea is to encode an ordered pair $(k,l)$, where $k,l<\lambda$, as $\setbr{\eta_k,l}$.  Clearly, for $k,k' < j < \lambda$ and $l,l' < \lambda$, we have 
\begin{equation}\label{eqn:ordpair}
  \vonset{\setbr{\eta_k,l}} \simi{\eta_j + 1} \vonset{\setbr{\eta_{k'},l'}} \iff k=k' \text{ and } l=l'
\end{equation}
For $j < \lambda$ and $c \in \lambda^{j}$, we set
\begin{displaymath}
\lbitsen{c} \eqdef \setbr{\setbr{\eta_k,c_k} \mid k < j} \cup \setbr{\setbr{\eta_j,l} \mid l < \lambda}
\end{displaymath}
For a set of ordinals $R$ we write $\vonset{R}$ for a parent of $(\von{j})_{j \in R}$, and for a set of sets of ordinals $S$ we write $\vonsetset{S}$ for a parent of $(\vonset{R})_{R \in S}$.  For $k \leqslant j < \lambda$ and $c \in \lambda^{k}$ and $d \in \lambda^{j}$, we deduce from~(\ref{eqn:ordpair}) that
\begin{equation}\label{eqn:lrestrict}
  c = d \restriction_k \iff \vonsetset{\lbitsen{c}} \simi{\eta_k + 2} \vonsetset{\lbitsen{d}}
\end{equation}
The cofinal embedding $\beta$ of the complete $\lambda$-ary $\lambda$-tree into $(\nucoalgpk{j})_{j < i}$ has index map $j \mapsto \eta_j + 2$.  The $j$-th level map $\betone \ccons \lambda^{j} \to \nuordpk{\eta_j+2}$ sends $c$ to $\pcote{\eta_j+2}{\vonsetset{\lbitsen{c}}}$.  Injectivity and naturality follow from~\ref{eqn:lrestrict}.

For case (\ref{item:notinacc}), Proposition~\ref{prop:nofull} gives a tidy $\lambda$-tree $D$, with all levels $\leqlambda$-sized, that has no full branch.  We give a cofinal embedding $\alpha$ of $D$ in the complete binary $\lambda$-tree as follows.  Since $\lambda$ is not strongly inaccessible, there is a cardinal $\mu < \lambda$ such that $2^{\mu} \geqslant \lambda$.  The index map is $j \mapsto \mu \times j$.  This is well-defined and cofinal because $j < \lambda$ implies $j \leqslant \mu \times j < \lambda$, since $1 \leqslant \mu < \lambda$ and $\lambda$ is regular.   We give the $j$-th level map $\alone{j} \ccons D_j \to \bits^{\mu \times j}$ by induction on $j$, ensuring naturality wrt all $k \leqslant j$.  For the successor case, we take an injection from the $\leqslant\!\! \lambda$-many $j+1$-developments of each $x \in D_j$ to the $2^{\mu}$-many $\mu \times (j+1)$ developments of $\alone{j}x \in \bits^{\mu \times j}$.  The case where $j$ is a limit is uniquely defined, as in the proof of Proposition~\ref{prop:allintocompbin}.

Next we obtain an embedding $\beta$ of the complete  binary $\lambda$-tree into $(\nucoalgpk{j})_{j < \lambda}$, with index map $j \mapsto j+1$, as in Corollary~\ref{cor:bintreeinto}. The $\beta\alpha$-image of $D$ is a channel through $(\nucoalgpk{j})_{j < \lambda}$, with all levels  $\leqlambda$-sized, that has a root but no full branch.

For case (\ref{item:nottree}), Proposition~\ref{prop:ghostnotree} gives a $\psetlambda$-ghost at $\lambda$, which is also a $\psetkappa$-ghost at $\lambda$.
\end{proof}

\subsection{Consequences of the ghosts}

Our aim is to revisit each of the properties listed in Section~\ref{sect:count}, to see at which ordinals it holds.  So we want to obtain many negative results from a ghost.  We begin with some methods for doing so.
\begin{lem} \label{lem:notinimage}
  A $\pset$-ghost at limit $i$ is not in the image of $\nuconn{i}{i+1} \pset$.
\end{lem}
\begin{proof}
Let $\nuconn{i}{i+1}b$, for $b \in \nuordp{i+1}$, be a $\pset$-ghost at $i$.  It corresponds across $\Psi_i$ to the range of $b$. Any element of $b$ would be a full branch of the range of $b$, which does not exist, so $b$ is empty.  But then its range is empty, contradiction. \end{proof}

We shall use the singleton operation $\setbr{-}$ applied $n$ times to a ghost.
\begin{lem} \label{lem:ghostn}
  Let $a$ a $\pset$-ghost at positive limit $i$, and $n \in \nats$.  Then $\setbr{-}^{n}a \in \nuordp{i+n}$ has the following properties.
  \begin{enumerate}
   \item \label{item:ghostnpcoconn} $\pcofc{i+n}{n+1}{\setbr{-}^{n}a} \not = \nuconn{n+1}{i+n}\setbr{-}^{n}a$. 
  \item \label{item:ghostnrange} $\setbr{-}^{n}a$ is not in the image of $\nuconn{i+n}{i+n+1} \pset$.
  \item \label{item:ghostnpset} $\setbr{-}^{n}a$ is not $\pset$-coalgebraic.
  \item \label{item:ghostnnotse} $\pfc{i+n}{\setbr{-}^{n}a} \sim \pfc{i+n}{ \setbr{-}^{n}0(i)}$ but $\setbr{-}^{n}a \not= \setbr{-}^{n}0(i)$.
  \item \label{item:ghostnnotsim} $\pfc{i+n}{\setbr{-}^{n}a} \not\simi{n+1} \pfc{j}{\nuconn{j}{i+n} \setbr{-}^{n}a}$ for any $j < i$ such that $j > n$.  
  \end{enumerate}
\end{lem}
\begin{proof} \hfill
\begin{enumerate}
\item Induction on $n$ gives $\pcofc{i+n}{n+1}{\setbr{-}^{n}a}) = \setbr{-}^{n} \emptyset$ and $\nuconn{n+1}{i+n}\setbr{-}^{n}a = \setbr{-}^{n} \setbr{\emptytup}$.
\item Induction on $n$. The case $n=0$ is Lemma~\ref{lem:notinimage}.  The inductive step follows from $\nuconn{i+n+1}{i+n+2}b = \setbr{\nuconn{i+n}{i+n+1}c \mid c \in b}$.
\item Follows from part (\ref{item:ghostnrange}).
\item Each part is by induction on $n$.
\item We have  $\nuconn{j}{i+n} \setbr{-}^{n}a \rightsquigarrow^{n} \nuconn{\predm{n}{j}}{i} a$.  Since $n < j < i$ we have $0 < \predm{n}{j} < i$, so $\nuconn{\predm{n}{j}}{i} a$ has a successor.  But if $\setbr{-}^{n}a \rightsquigarrow^{n} b$, then $b = a$, which has no successor. 
\end{enumerate}
\end{proof}

For our results, we use the following terminology.
\begin{defi}
  For an ordinal $i$, let $i = i'+m$ where $i'$ is a limit and $m \in \nats$.  Then we write
  \begin{spaceout}{rcll}
    \limp{i} & \eqdef & i', & \text{the \emph{limit part} of $i$.} \\
    \nump{i} & \eqdef & m, & \text{the \emph{finite part} of $i$.}
  \end{spaceout}%
\end{defi}

The following completes Lemma~\ref{lem:pconuconn}.
\begin{thm} \label{prop:genpcoconn} 
For $j \leqslant i$, the following are equivalent.
\begin{enumerate}
\item \label{item:genpcoconnp} For all  $a \in \nuordpk{i}$ we have $\pcop{i}{a} = \nuconn{j}{i}a$.
\item \label{item:genpcoconnc} Either $j \leqslant \nump{i}$ or $\limp{i}$ is $\kappa$-extensible.
\end{enumerate}
\end{thm}
\begin{proof}
For (\ref{item:genpcoconnc})$\Rightarrow$(\ref{item:genpcoconnp}), the case ($j \leqslant \nump{i}$) is Proposition~\ref{prop:finpost}(\ref{item:finpost}), and the case of $\kappa$-extensible $\limp{i}$ is proved the same way as Lemma~\ref{lem:pconuconn}.

For the converse, let $a$ be a $\psetkappa$-ghost $a$ at $\limp{i}$.  If  $\pcofc{i}{j}{\setbr{-}^{\nump{i}}a} = \nuconn{j}{i}{\setbr{-}^{\nump{i}}a}$ then (since $\nump{i}+1 \leqslant j$) we have $\pcofc{i}{\nump{i}+1}{ \setbr{-}^{\nump{i}}a} = \nuconn{\nump{i}+1}{i}{\setbr{-}^{\nump{i}}a}$, contradicting Lemma~\ref{lem:ghostn}(\ref{item:ghostnpcoconn}).
\end{proof}

We come to the main result of the paper, which completes Proposition~\ref{prop:countord}.
\begin{thm} \label{prop:finalcover}
For any ordinal $i$, the following are equivalent.
  \begin{enumerate}
  \item \label{item:finalcoverpself} For all $a \in \nuordpk{i}$ we have $\pcofc{i}{i}{a} = a$.
  \item \label{item:finalcoverco} Every $a \in \nuordpk{i}$ is $\pset$-coalgebraic.
  \item \label{item:finalcoversimi} For $a,b \in \nuordpk{i}$, if $a \simi{i} b$ then $a=b$.
  \item \label{item:finalcoversim} For $a,b \in \nuordpk{i}$, if $a \sim b$ then $a=b$.
  \item \label{item:finalcoverc} $\limp{i}$ is $\kappa$-extensible.
  \end{enumerate}
\end{thm}
\begin{proof}
Proposition~\ref{prop:genpcoconn} gives (\ref{item:finalcoverc})$\Leftrightarrow$(\ref{item:finalcoverpself}).  Parts~(\ref{item:finalcoverco})--(\ref{item:finalcoversim}) follow just as in Proposition~\ref{prop:countord}.  

If $\limp{i}$ is not $\kappa$-extensible, let $a$ be a $\psetkappa$-ghost at $\limp{i}$.  By Lemma~\ref{lem:ghostn}(\ref{item:ghostnpset}), $\setbr{-}^{\nump{i}} a$ is not $\pset$-coalgebraic, and is bisimilar to but distinct from  $0(i)$.
\end{proof}

Another way of obtaining negative results from a ghost is the following.
\begin{lem} \label{lem:homeghost}
Let $i < \kappa$ be a positive limit, and let $a$ be a $\psetkappa$-ghost at $i$.  There is a sequence $(x_m)_{m \in \nats}$ of $\lkappa$-branching pointed systems such that
\begin{itemize}
  \item $\pcop{i}{x_0} \rightsquigarrow a$
\item every successor of $x_0$ has a successor
 \item $x_{m+1}$ is a parent of just $x_m$.
\end{itemize}
\end{lem}
\begin{proof}
The only difficulty is to obtain $x_0$, for we can then obtain $x_m$ by induction on $m$.  Since $a$ is $\psetkappa$-Cauchy, we may choose, for each positive $j < i$, a $\lkappa$-branching pointed system $y_j$ such that $a_{j+1} = \pcop{j+1}{y_j}$ (and hence  $a_{j} = \pcop{j}{y_j}$).  Let $x_0$ be a parent of $(y_j)_{j < i}$.  Then for $j < i$, we have $(\pcop{i}{x_0})_{j+1}  =  \pcop{j+1}{x_0} 
  =  \setbr{\pcop{j}{y} \mid x_0 \rightsquigarrow y} 
  \ni  a_{j}$, 
so $\pcop{i}{x_0} \rightsquigarrow a$.  Any successor $y$ of $x_0$ is, for some $j < i$, an embedding applied to $y_j$.  Since $\pcop{j+1}{y_j}$ is inhabited, $y_j$ and hence $y$ must have a successor.\end{proof}

We shall complete Corollary~\ref{cor:countchar}, which characterizes $\pcop{i}{x}$ and $\nuconn{j}{i}a$, in two parts: firstly considering when these elements have the required property (Proposition~\ref{prop:connprojsim}), and secondly considering when no other element has it (Proposition~\ref{prop:char}).
\begin{prop} \label{prop:connprojsim} \hfill
 \begin{enumerate}
  \item \label{item:gensimproj} For an ordinal $j$, the following are equivalent.
    \begin{enumerate}
    \item \label{item:gensimprojp} For any $\lkappa$-branching pointed system $x$ we have $x \simi{j} \pcop{j}{x}$.
   \item \label{item:gensimprojc} Either $j \geqslant \kappa$ (for $\kappa < \biginfty$) or  $\limp{j}$ is $\kappa$-extensible.
    \end{enumerate}
\item \label{item:gensimconn} For ordinals $j \leqslant i$, the following are equivalent.
  \begin{enumerate}
  \item \label{item:gensimconnp} For any $a \in \nuordpk{i}$, we have $\nuconn{j}{i} a \simi{j} a$.
  \item \label{item:gensimconnc} Either $j=i$ or $j \leqslant \nump{i}$ or $j \geqslant \kappa$ (for $\kappa < \biginfty$) or both $\limp{j}$ and $\limp{i}$ are $\kappa$-extensible.
  \end{enumerate}
  \end{enumerate}
\end{prop}
\begin{proof} For (\ref{item:gensimprojc})$\Rightarrow$(\ref{item:gensimprojp}), the $\kappa$-extensible case is proved as in Corollary~\ref{cor:countchar}(\ref{item:countcharproj} and the case $j \geqslant \kappa$ holds by Proposition~\ref{prop:succpcopk}.  For (\ref{item:gensimconnc})$\Rightarrow$(\ref{item:gensimconnp}), the $\kappa$-extensible case is proved as in Corollary~\ref{cor:countchar}(\ref{item:countcharconn}, the case $j \geqslant \kappa$ is by Proposition~\ref{prop:embedtrans} and the case $j \leqslant \nump{i}$ follows from Proposition~\ref{prop:finpost}(\ref{item:finpost}).

For the converse we proceed as follows.
\begin{enumerate}
\item If $j < \kappa$ and $\limp{j}$ is not $\kappa$-extensible, let $a$ be a $\psetkappa$-ghost at $\limp{j}$ and form the sequence $(x_{m})_{m \in \nats}$ as in Lemma~\ref{lem:homeghost}.   Put $n \eqdef \nump{j}$.  Then
  $\pcop{j}{x_{n}} \rightsquigarrow^{n} \pcop{j}{x_0} \rightsquigarrow a$ and $a$ has no
  successor, whereas if $x_{n} \rightsquigarrow^{n+1}y$ then $y$ has a successor.  So $\pcop{j}{x_{n}} \not \simi{n+2} x_n$.
% If $i+n$,
%   where $i$ is a limit and $n \in \nats$, is not finally
%   $\kappa$-covered and $i < \kappa$, let $a$ be a $\psetkappa$-ghost
%   at $i$. Since $a$ is $\psetkappa$-Cauchy, for each $j < i$ put we
%   have $a_{j+1} = \pcot{M_j}{j+1}{x_j}$ for some $\lkappa$-branching
%   pointed system $(M_j,x_j)$.  Let the transition system
%   $M$ be $\sum_{j < i}M_j$ with extra elements
%   $\setbr{s_m \mid m \in \nats}$ where $s_0$ has successor set
%   $\setbr{\mathsf{in}_j x_j \mid j < i}$ and $s_{m+1}$ has sole
%   successor $s_m$.  Then
%   $b \eqdef \pcot{M}{i}{s_0} \rightsquigarrow a$ and
%   $\pcot{M}{i+n}{s_n} = \setbr{-}^{n}b$.  Thus
%   $\pcot{M}{i+n}{s_{n}} \rightsquigarrow^{n+1} a$ and $a$ has no
%   successor, whereas if $s_{n} \rightsquigarrow^{n+1}x$ then
%   $x = \in_j x_j$ for some $j<i$, and $x_j$ (and hence $x$) has a
%   successor.  So $\pcot{M}{i+n}{s_n} \not \simi{n+2} (M,s_n)$ and
%   therefore $\pcot{M}{i+n}{s_n} \not \simi{n+2} (M,s_n)$.
\item Suppose $\nump{i} < j < i,\kappa$, and $\limp{j}$ and $\limp{i}$ are not both $\kappa$-extensible.  We consider three cases.
  \begin{itemize}
\item If $\limp{j} = \limp{i}$, put $L = \limp{j}$ and $j = L+n$ and $i = L+m+n+1$.   Let $a$ be a  $\psetkappa$-ghost at $L$ and form the sequence $(x_m)_{m \in \nats}$ as in Lemma~\ref{lem:homeghost}.  Put $b \eqdef \pcop{i}{x_n}$.  Then $\nuconn{j}{i}b = \pcop{L+n}{x_n} \rightsquigarrow^{n+1}a$, which has no successor.  On the other hand, if $b \rightsquigarrow^{n+1} c$ then $c = \pcop{L+m}{y}$, for some successor $y$ of $x_0$, and since $y$ has a successor, $c$ does too.  So $b \not \simi{n+2} \nuconn{j}{i}b$.
  \item If $\limp{i}$ is not $\kappa$-extensible and $j < \limp{i}$, let $a$ be a $\psetkappa$-ghost at $\limp{i}$.  Put $n \eqdef \nump{i}$ and apply Lemma~\ref{lem:ghostn}(\ref{item:ghostnnotsim}).
 % $a^n \eqdef \setbr{-}^{n} a$.  Then $\nuconn{j}{i}a^n \rightsquigarrow^{n} \nuconn{\predm{n}{j}}{\limp{i}} a$.  Since $\nump{j} < j < \limp{i}$ we have $0 < \predm{n}{j} < \limp{i}$, so $\nuconn{\predm{n}{j}}{\limp{i}} a$ has a successor.  But if $a_n \rightsquigarrow^{n} c$, then $c = a$, which has no successor.  So $\pfc{j}{\nuconn{j}{i} a_n} \not \simi{n+1} \pfc{i}{a_n}$.
  \item If $\limp{i}$ is $\kappa$-extensible then $\limp{j}$ is not, so let $a$ be a $\psetkappa$-ghost at $\limp{j}$ and form the  sequence $(x_m)_{m \in \nats}$ as in Lemma~\ref{lem:homeghost}.   Put $n \eqdef \nump{j}$ and $b \eqdef \pcop{i}{x_n}$.   Then $b \simi{i} x_n$ but $\nuconn{j}{i} b = \pcop{j}{x_n} \not \simi{n+2} x_n$, as we saw in the proof of part~\ref{item:gensimproj}.  So $b \not \simi{n+1} \nuconn{j}{i} b$. \qedhere
   \end{itemize}
 \end{enumerate}
 \end{proof}

\begin{prop} \hfill \label{prop:char}
 \begin{enumerate}
  \item \label{item:charproj} For any ordinal $j$, the following are equivalent.
    \begin{enumerate}
   \item \label{item:charprojpo} For any $\lkappa$-branching pointed system $x$, any $b \in \nuordpk{j}$ such that $b \simi{j} x$ is $\pcop{j}{x}$.
    \item \label{item:charprojpt} For any $\lkappa$-branching pointed system $x$, we can characterize $\pcop{j}{x}$ as the unique $b \in \nuordpk{j}$ such that $b \simi{j} x$.
  \item \label{item:charprojc} $\limp{j}$ is $\kappa$-extensible.
    \end{enumerate}
  \item \label{item:charconn} For any ordinals $j \leqslant i$, the following are equivalent.
    \begin{enumerate}
    \item \label{item:charconnpo} For any $a \in \nuordpk{i}$, any $b\in \nuordpk{j}$ such that $b \simi{j} a$ is $\nuconn{j}{i}a$.
    \item \label{item:charconnpt} For any $a \in \nuordpk{i}$, we can characterize $\nuconn{j}{i} a$ as the unique $b\in \nuordpk{j}$ such that $b \simi{j} a$.
   \item \label{item:charconnc} Either $j \leqslant \nump{i}$ or both $\limp{j}$ and $\limp{i}$ are $\kappa$-extensible.
    \end{enumerate}
  \end{enumerate}
\end{prop}
\begin{proof} The proof of~(\ref{item:charprojc})$\Rightarrow$(\ref{item:charprojpt}) and~(\ref{item:charconnc})$\Rightarrow$(\ref{item:charconnpt}) is proved the same way as Corollary~\ref{cor:countchar}.  For the converse parts, we form a sequence  $(z_m)_{m \in \nats}$ where $z_0$ is a parent of nothing, i.e.\ a pointed system with no successor, and $z_{m+1}$ a parent of just $z_m$.  (Essentially this is the Zermelo encoding of the natural numbers.)  

In both parts, if $\limp{j}$ is not $\kappa$-extensible, let $a$ be a
   $\psetkappa$-ghost at $\limp{j}$ and put $n \eqdef \nump{j}$.  Then
   $\pcop{j}{z_n} = \setbr{-}^{n}\zat{\limp{j}}$ (lemma: $\pcop{k+n}{z_n} = \setbr{-}^{n}\zat{k}$, by induction on $n$), which is bisimilar to
   but distinct from $\setbr{-}^{n}a$, so
(\ref{item:charprojpo}) is false. And, in part~\ref{item:charconn}, this is also $\nuconn{j}{i}{\pcop{i}{z_n}}$, so~(\ref{item:charconnpo}) is false.  

In part~(\ref{item:charconn}), if $\limp{i}$ is not $\kappa$-extensible and $j > m \eqdef \nump{i}$, let $a$ be a $\psetkappa$-ghost at $\limp{i}$.  Then $\pcop{j}{z_m}$ is bisimilar to $z_m$ (this is a lemma for all $j>m$, by induction on $m$) and hence to $\setbr{-}^{m}a$.  But $\pcop{j}{z_m} = \nuconn{j}{i}\setbr{-}^{m}a$ would imply, by applying $\nuconn{m+1}{j}$, that $\pcop{m+1}{z_m} = \nuconn{m+1}{i}\setbr{-}^{m}a$ i.e.\ $\setbr{-}^{m} = \setbr{-}^{m} \emptytup$, contradiction.  So~(\ref{item:charconnpo}) is false. \end{proof}

%Immediate from Propositions~\ref{prop:connprojsim}--\ref{prop:simthen}.\end{proof}

We complete Proposition~\ref{prop:brdepth} as follows.
\begin{prop} \label{prop:brdepthbeyond}
Assume $\kappa < \biginfty$.  
  \begin{enumerate}
  \item 
% For any ordinal $j$, the following are equivalent.
%     \begin{itemize}
%     \item  Every element of $\nuordpk{\kappa+j}$ that is is $\lkappa$-branching is $\psetkappa$-coalgebraic.
%     \item Either $j \geqslant \omega$ or $\kappa$ has the tree property.
%     \end{itemize}
% \item For ordinals $j \leqslant i$, the following are equivalent.
%   \begin{itemize}
%   \item Every element of $\nuordpk{\kappa+j}$ that is $\lkappa$-branching at depth $<i$ is in the range of $\nuconn{\kappa+j}{\kappa+i}$.
%   \item 
%   \end{itemize}
Let $j$ be an ordinal.   Every element of $\nuordpk{\kappa+j}$ that is $\psetkappa$-coalgebraic is $\lkappa$-branching, and conversely iff either $j \geqslant \omega$ or $\kappa$ has the tree property .
  \item Let $j\leqslant i$ be ordinals.  Every element of $\nuordpk{\kappa+j}$ that is in the range of $\nuconn{\kappa+j}{\kappa+i}$ is $\lkappa$-branching at depth $<i$, and conversely iff either $j=i$ or $j \geqslant \omega$ or $\kappa$ has the tree property.
  \end{enumerate}
\end{prop}
\begin{proof} We only need to prove the converse parts.  If $j \geqslant \omega$ then every element of $\nuordpk{\kappa+j}$ is $\lkappa$-branching.  If $\kappa$ has the tree property, we proceed as for $\aleph_0$.  Otherwise, let $a$ be a $\psetkappa$-ghost at $\kappa$.   The element $\setbr{-}^{j}a$ is $\lkappa$-branching but not $\psetkappa$-coalegebraic and not in the image of $\nuconn{\kappa+j}{\kappa+j+1}$, by Lemma~\ref{lem:ghostn}(\ref{item:ghostnrange})--(\ref{item:ghostnpset}). \end{proof}

\subsection{Surjectivity and $\psetkappa$-coalgebraicity}

We next consider which connecting maps are surjective.
\begin{prop} \label{prop:treepropsurj}
Let $\lambda$ be a regular infinite cardinal with the tree property that is less than $\kappa$.  Then $\psetkappa$ and $\psetpluskappa$ preserve, up to a surjection, the limit of any inverse $\lambda$-chain $D$ with all levels $\llambda$-sized.  Explicitly, any channel $C$ through $D$ is the range of a $\lkappa$-sized set of full branches. 
\end{prop}
\begin{proof}
It suffices to prove the case where $D_i = \lim_{j < i} D_j$ for each limit $i<\lambda$.  For $i < \lambda$, each $x \in C_i$ extends by the tree property to a full branch through $D$, because all levels are $\llambda$-sized.  For each $i < \lambda$ and $x \in C_i$,   choose such a branch $\theta_i(x)$, by the Axiom of Choice.  The set  $\setbr{\theta_i(x) \mid i < \lambda, x \in C_{i}}$ has size $\leqslant \lambda \times \lambda = \lambda < \kappa$ and has range $C$. \end{proof}

\begin{prop} \label{prop:surject}
  For ordinals  $j \leqslant i$, the following are equivalent.
\begin{enumerate}
\item \label{item:surjectp} The connecting map $\nuconn{j}{i}\psetkappa$ is an surjection.
\item \label{item:surjectc} Either $j=i$ or $\limp{j}$ is either $0$ or $\omega$-cofinal or weakly compact or $\geqslant \kappa + \omega$ (for  $\kappa < \biginfty$).
\end{enumerate}
\end{prop}
\begin{proof}
By Proposition~\ref{prop:projconnect} we may assume $i=j+1$.  The case $j \geqslant \kappa$ has been treated (Proposition~\ref{prop:injective} and~\ref{prop:noearlier}), so we assume $j<\kappa$.  Since $\psetkappa$ preserves and reflects surjectivity, we may assume $j$ is a limit.

For (\ref{item:surjectc})$\Rightarrow$(\ref{item:surjectp}), it suffices by Proposition~\ref{prop:psi} for $\psetkappa$ to preserve the limit of any inverse $j$-chain up to a surjection.   The $j=0$ case is trivial, the $\omega$-cofinal case is by Proposition~\ref{prop:pressurj}, and the weakly compact case is by Proposition~\ref{prop:treepropsurj}.  For the converse, if $j < \kappa$ and $j$ is not $\kappa$-extensible then there is a $\psetkappa$-ghost at $j$, which is not even $\pset$-coalgebraic. \end{proof}

\begin{cor} \label{prop:allpsetkco}
  For an ordinal $i$, the following are equivalent.
  \begin{itemize}
  \item \label{item:allpsetkcop} Every element of $\nuordpk{i}$ is $\psetkappa$-coalgebraic. 
  \item \label{item:allpsetkcoc}  $\limp{i}$ is either $0$ or $\omega$-cofinal or weakly compact or $\geqslant \kappa + \omega$ (for $\kappa < \biginfty$).
  \end{itemize}
\end{cor}
\begin{proof} Follows from Proposition~\ref{prop:surject} by Proposition~\ref{prop:projconnect}. \end{proof}

We next see that $\psetkappa$-Cauchy is weaker than $\psetkappa$-coalgebraic.
\begin{prop} \label{prop:cauchyco}
  For a limit ordinal $i$, the following are equivalent.
  \begin{enumerate}
%  \item Every element of $\nuordpk{i}$ is $\psetkappa$-coalgebraic.
  \item \label{item:cauchycop} Every $\psetkappa$-Cauchy element of $\nuordpk{i}$ is $\psetkappa$-coalgebraic. 
  \item \label{item:cauchycoc}  $i$ is either $0$ or $\omega$-cofinal or weakly compact or $\geqslant \kappa + \omega$ (for $\kappa < \biginfty$).
  \end{enumerate}
\end{prop}
\begin{proof} For (\ref{item:cauchycoc})$\Rightarrow$(\ref{item:cauchycop}), every element of $\nuordpk{i}$ is $\psetkappa$-coalgebraic.  For the converse: if $i < \kappa$ and is neither $0$ nor $\omega$-cofinal nor weakly compact, let $a$ be a $\psetkappa$-ghost at $i$.  It is $\psetkappa$-Cauchy but not $\pset$-coalgebraic.  For $i = \kappa$, Proposition~\ref{prop:kappasucc} gives a $\psetkappa$-Cauchy element of $\nuordpk{\kappa}$ that has $\kappa$ successors and so is not $\psetkappa$-coalgebraic. \end{proof}

We substantiate our claim in Section~\ref{sect:transrestr} that the coalgebraic elements need not form a subsystem of the final chain.
\begin{prop} \label{prop:cosub}
  For a limit ordinal $i$, the following are equivalent.
  \begin{enumerate}
  \item \label{item:cosubp} The $\psetkappa$-coalgebraic elements form a subsystem of $\nuordpk{i}$.
  \item \label{item:cosubc}  $i$ is either $0$ or $\omega$-cofinal or weakly compact or $\geqslant \kappa$ (for $\kappa < \biginfty$).
  \end{enumerate}
\end{prop}
\begin{proof} For (\ref{item:cosubc})$\Leftarrow$(\ref{item:cosubp}), if $i \geqslant \kappa$, we apply Proposition~\ref{prop:succpcopk}.  In the other cases, all elements are $\psetkappa$-coalgebraic.  For the converse, if the conclusion is false let $a$ be a $\psetkappa$-ghost at $i$.  It is not $\pset$-coalgebraic, and by Lemma~\ref{lem:homeghost} it is a successor of a $\psetkappa$-coalgebraic element. \end{proof}

\section{Related Work} \label{sect:related}
\subsection{Trees vs tidy trees} \label{sect:treestidy}

Our formulation of the Aronszajn and Kurepa properties used the notion of \emph{tidy} tree.  This suited our purposes, but to avoid confusion it must be compared with the following more general notion that commonly appears in the set-theoretic literature.  
\begin{defi}
  Let $I$ be a well-ordered set.  An \emph{$I$-tree} is an inverse $I$-chain where every level is inhabited. 
\end{defi}
If $I$ has a least element, then clearly every tidy $I$-tree is an $I$-tree; but there is also a kind of converse, as follows.  (Cf.~\cite{Jech:settheory}[Lemma 9.13].)
\begin{prop} \label{prop:tidynoprob}
Let $\lambda$ be a regular infinite cardinal.  For any $\lambda$-tree $D$ with all levels $\llambda$-sized, there is a tidy $\lambda$-tree $E$ with all levels $\llambda$-sized and a bijection from the full branches of $E$ to those of $D$.
\end{prop}
\begin{proof}
We first construct a channel through $D$, with all levels inhabited, that contains all the full branches.  For each $i < \lambda$, let $C_i$ be the set of all $x \in D_i$ that have a $k$-development for all $k < \lambda$ such that $k \geqslant i$.   To see that $C$ is a channel, let $j \leqslant i < \lambda$ and $x \in C_j$.  For each $k < \lambda$ such that $k \geqslant i$, let $R(k)$ be the set of $i$-developments of $x$ that has an $k$-development in $C$.  Since $C_i$ is $\llambda$-sized the subset $\bigcap_{i \leqslant k < \lambda} R(k)$ has an element, which is an $i$-development of $x$ in $C_i$.  We likewise prove for all $i < \lambda$ that $C_i$ is inhabited.  Evidently every full branch of $D$ is a full branch of $C$.

Next we form the inverse chain $E$, whose $i$th level, for $i < \lambda$, is $\lim_{j < i}D_j$, and whose connecting maps are given by restriction.  Let $\alpha \colon D \to E$ be the cofinal embedding with index map $i \mapsto i+1$ and $i$-th level map sending $x \in D_i$ to $(\ddiag{j}{i}x)_{j \leqslant i} \in E_{i + 1}$. The $\alpha$-image of $C$ is a tidy tree, and by Proposition~\ref{prop:cofinalemb} has the required properties.
\end{proof}
\begin{cor}\hfill
  \begin{enumerate}
  \item There is a Kurepa tree (an $\omega_1$-tree with all levels countable that has more than $\aleph_1$ full branches) iff there is a tidy one.
  \item Let $\lambda$ be a regular infinite cardinal.  There is a $\lambda$-Aronszajn tree (a $\lambda$-tree with all levels $\llambda$-sized that has no full branch) iff there is a tidy one.
  \end{enumerate}
\end{cor}
So our use of tidy trees is not a significant change from the usual formulation.

\subsection{Saturation}

In~\cite{AdamekLevyMiliusMossSousa:saturate}, an account is given of the final chain of $\pset$ at countable ordinals.  Elements of the final chain are observed to be \emph{saturated}.  At limit ordinals, this property may be expressed as follows.
\begin{prop} \label{prop:saturate}
Let $i$ be a $\kappa$-extensible (e.g.\ countable) limit.  For any $a \in \nuordp{i}$ and pointed system $x$, if for all $j < i$, $a$ has a $\fctrans{i}$-successor $y$ such that $y \simi{j} x$, then $a$ has a $\fctrans{i}$-successor $y$ such that $y \simi{i} x$.
\end{prop}
\begin{proof} 
By Theorem~\ref{prop:finalcover}, $a = \pcop{i}a$, and transitions to $\pcop{i}{x}$ by  Proposition~\ref{prop:pcoppcop} since $\pred{i}=i$.  Lastly, $\pcop{i}{x} \simi{i} x$ by
  Proposition~\ref{prop:connprojsim}(\ref{item:gensimproj}). \end{proof}
  % \begin{spaceout}{rcll}
  %   a & = & \pcofc{i}{i}a & \text{by Proposition~\ref{prop:finalcover}} \\
  %   & \rightsquigarrow & \pcot{M}{i}{x} &  \text{by Proposition~\ref{prop:projmapst} and $\pred{i}=i$} \\
  % \end{spaceout}%
  % and $\pfc{i}{\pcot{M}{i}{x}} \simi{i} (M,x)$ by
  % Proposition~\ref{prop:connprojsim}(\ref{item:gensimproj}).
It is not known whether this property holds in the case $i = \omega_1$.  We conjecture that it does not, even in the case where $a \in \nuordpc{\omega_1}$ and $x$ is countably branching.

\subsection{Cauchy completeness} \label{sect:cauchyco}

To enable comparison with~\cite{FortiHonsell:modselfdescript,LazicRoscoe:nonwellfound}, let us formulate Proposition~\ref{prop:cauchyco} in terms of Cauchy sequences.  
\begin{itemize}
% \item An $i$-sequence of pointed systems $(x_j)_{j < i}$ is \emph{$\lkappa$-branching} when  $x_j$ is $\lkappa$-branching for all $j < i$.
% \item Two $i$-sequences of pointed systems $(x_j)_{j < i}$ and  $(y_j)_{j < i}$ are \emph{equivalent} when $x_j \simi{j} y_j$  for all $j<i$.
\item An $i$-sequence of pointed systems $(x_j)_{j < i}$ is
  \emph{Cauchy} when for all $j \leqslant k< i$ we have $x_j \simi{j} x_k$.
\item A \emph{limit} for a Cauchy $i$-sequence $(x_j)_{j < i}$ is a pointed system $y$ such that for all $j<i$ we have $x_j \simi{i} y$.  
\end{itemize}
We then have the following.
\begin{prop} \label{prop:cauchyseq}
 For an ordinal $i$, the following are equivalent.
 \begin{enumerate}
 \item \label{item:cauchyseqp} Every Cauchy $i$-sequence of $\lkappa$-branching pointed systems has a $\lkappa$-branching limit.
 \item \label{item:cauchyseqc} $i$ is either a successor or $0$ or $\omega$-cofinal or weakly compact or $\geqslant \kappa + \omega$ (for $\kappa < \biginfty$).
 \end{enumerate}
\end{prop}
\begin{proof} The successor case is evident.  For the limit case, 
(\ref{item:cauchyseqp}) corresponds to Proposition~\ref{prop:cauchyco}(\ref{item:cauchycop}).  To see this, note that a Cauchy sequence $(x_j)_{j < i}$ gives a $\psetkappa$-Cauchy element $a = (\pcop{j}{x})_{j < i} \in \nuordpk{i}$, and every $\psetkappa$-Cauchy element arises in this way.  Moreover, $y$ is a limit for $(x_j)_{j < i}$ iff $a= \pcop{i}{y}$, so $(x_j)_{j < i}$ has a limit iff $a$ is $\psetkappa$-coalgebraic.  \end{proof}

\subsection{Comparison}

As stated, the work of~\cite{FortiHonsell:modselfdescript} provides the basis of the method we have used to obtain ghosts.  Because that paper is concerned with models of set theories, it treats primarily the full powerset functor.  It studies systems that resemble the final chain system in satisfying Proposition~\ref{prop:pcoppcop}, but differ from it by including only coalgebraic elements.  As we have seen---Proposition~\ref{prop:cosub}---these do not form a subsystem of the final chain, so the difference is considerable.

The work of~\cite{LazicRoscoe:nonwellfound} treats not only systems of coalgebraic elements but also systems of Cauchy elements (and the latter do form a subsystem of the final chain), specifically for the full powerset functor.  A more elaborate proof of Cauchy incompleteness is given in order to obtain additional negative results.  Relating that work to the present paper is a matter for future research.

\section{Conclusions} \label{sect:conclusions}

We have investigated several properties that the final chain of $\psetkappa$, viewed as a transition system, enjoys at countable ordinals.  In particular, the \emph{strong extensionality} property: bisimilar elements are equal.  We have seen that these properties do not hold at $\omega_1$, because of an element distinct from $\zat{\omega_1}$ that has no successor---a ``ghost''.  Using the same  method, we have precisely identified those ordinals at which each of these properties does hold.

\newcommand{\etalchar}[1]{$^{#1}$}

%\bibliography{mybib}

\end{document}